%% file: DW3-11-arxiv-2.tex
\documentclass[twocolumn]{revtex4}
\usepackage{tikz,graphicx,amsmath,amssymb,amsthm,thm-restate}
\usepackage{bbold,revsymb,latexsym,bbm}
\usepackage{caption}

\theoremstyle{plain}
\newtheorem{theorem}{Theorem}

\newtheorem{proposition}[theorem]{Proposition}

\theoremstyle{definition}
\newtheorem{definition}[theorem]{Definition}
\newtheorem{example}[theorem]{Example}

\newtheorem{corollary}[theorem]{Corollary}

\theoremstyle{remark}

\newlength{\blank}
\newcommand{\nc}{\newcommand}

\def\mc{\mathcal}

\def\R{\right}

\def\ox{\otimes}

\nc{\nts}[1]{{\sf \textcolor{red}{[#1]}}\textcolor{red}{\marginpar[\hfill !!!]{!!!}}}
\nc{\bra}[1]{\langle#1|}
\nc{\ket}[1]{|#1\rangle}
\nc{\ketbra}[2]{|#1\rangle\!\langle#2|}
\nc{\braket}[2]{\langle#1|#2\rangle}
\nc{\innerproduct}[2]{\langle #1, #2 \rangle}
\nc{\eref}[1]{(\ref{#1})}

\nc{\pure}[1]{\ketbra{#1}{#1}}
\nc{\partrans}[2]{\Gamma_{#1}\left(#2\right)}

\def\squareforqed{\hbox{\rlap{$\sqcap$}$\sqcup$}}
\def\qed{\ifmmode\squareforqed\else{\unskip\nobreak\hfil

\penalty50\hskip1em\null\nobreak\hfil\squareforqed
\parfillskip=0pt\finalhyphendemerits=0\endgraf}\fi}
\def\endenv{\ifmmode\;\else{\unskip\nobreak\hfil
\penalty50\hskip1em\null\nobreak\hfil\;
\parfillskip=0pt\finalhyphendemerits=0\endgraf}\fi}

\nc\supp{\mathrm{supp}}

\nc{\mathword}[1]{\mathinner{\mathrm{#1}}}

\def\tr{\mathword{Tr}}

\def\Hs{\mc{H}}

\DeclareRobustCommand\idop{\leavevmode\hbox{\small1\normalsize\kern-.33em1}}
\DeclareRobustCommand\scriptidop{\leavevmode\hbox{\fontsize{7}{8}\selectfont 1\scriptsize\kern-.33em1}}

\def\1{\idop}
\def\r{\rho}

\usepackage{enumerate}
\usepackage{url}
\usepackage{wrapfig}

\def\<{\leftarrow}

\def\1{{\mathbb 1}}
\def\rm{\mathrm}
\def\mc{\mathcal}

\def\t{\mathbf{t}}
\def\F{\mathbf{F}}

\def\id{\rm{id}}

\def\T{\mathrm{T}}

\def\ops{\mathbf{ops}}

\newcommand{\bs}[1]{\mathrm{#1}} %big system label
\def\q{\bs{Q}}
\def\Q{\bs{Q}}

\def\C{\bs{C}}
\def\r{\bs{R}}
\def\R{\bs{R}}
\def\P{\bs{P}}

\def\P{\bs{P}}

\def\Hs{\mc{H}}

\def\sdpop{\Lambda}
\def\sdpopQ{\Gamma}

\def\ai{\bs{A}}
\def\aic{\bs{\tilde{A}}}
\def\ao{\bs{A}'}
\def\aoc{\bs{\tilde{A}}'}
\def\aocc{\bs{\mathring{A}'}}
\def\bi{\bs{B}}
\def\bo{\bs{B}'}

\def\distop{\mc{Y}_{\aic\bo\<\aoc\bi}}
\def\PPTf{F_{\Gamma}}
\def\PPTD{D_{\Gamma}}

\def\a{\bs{a}}
\def\b{\bs{b}}

\def\b{\bs{b}}
\def\bi{\bs{B}}

\def\Zc{\mathcal{Z}}

\def\Alice{Alice}
\def\Bob{Bob}

\def\cdp{\mathcal{C}}

\def\UA{\mathbf{UA}}
\def\EA{\mathbf{EA}}
\def\NS{\mathbf{NS}}
\def\FHA{\mathbf{FHA}}
\def\FCA{\mathbf{FCA}}
\def\PPTp{\mathbf{PPTp}}

\def\sym{\mathbb{S}}
\def\asym{\mathbb{A}}

\def\slo{0.22}

\def\dM{K} %the dimension of the message system
\begin{document}
    
\date{August 18, 2014}

    \title{On the power of PPT-preserving and non-signalling codes}
    \author{Debbie Leung, William Matthews
    \thanks{William Matthews (will@northala.net) is with the Department of Applied
Mathematics and Theoretical Physics, University of Cambridge,
Cambridge, U.K. This work was initiated while W.M. was with the
Institute for Quantum Computing, University of Waterloo, Waterloo, ON,
Canada and completed while he was with the Statistical Laboratory at
the Department of Pure Mathematics and Mathematical Physics,
University of Cambridge.  Part of this work was conducted when
W.M. and Debbie Leung were participants of the programme "Mathematical
Challenges in Quantum Information" at the Isaac Newton Institute for
Mathematical Sciences.  W.M. gratefully acknowledges the support of
the NSERC, QuantumWorks and the Isaac Newton Trust. Debbie Leung is
with the Institute for Quantum Computing, and the Department of
Combinatorics and Optimization, University of Waterloo,
Ontario, Canada.  She is funded by the Canada Research Chair,
Canadian Institute For Advanced Research, and NSERC.} }

   %%% \maketitle GOES HERE FOR IEEEtran class
    
    \begin{abstract}
We derive ‘one-shot’ upper bounds for quantum noisy channel codes. We do so by regarding a channel code as a bipartite operation with an encoder belonging to the sender and a decoder belonging to the receiver, and imposing constraints on the bipartite operation. We investigate the power of codes whose bipartite operation is non-signalling from Alice to Bob, positive-partial transpose (PPT) preserving, or both, and derive a simple semidefinite program for the achievable entanglement fidelity. Using the semidefinite program, we show that the non-signalling assisted quantum capacity for memoryless channels is equal to the entanglement-assisted capacity. We also relate our PPT-preserving codes and the PPT-preserving entanglement distillation protocols studied by Rains. Applying these results to a concrete example, the $3$-dimensional Werner-Holevo channel, we find that codes that are non-signalling and PPT-preserving can be strictly less powerful than codes satisfying either one of the constraints, and therefore provide a tighter bound for unassisted codes. Furthermore, PPT-preserving non-signalling codes can send one qubit perfectly over two uses of the channel, which has no quantum capacity. We discuss whether this can be interpreted as a form of superactivation of quantum capacity.

    \end{abstract}
    
    \maketitle %%%% GOES HERE FOR revtex4 class (See above)
    
    \section{Introduction}
    A basic problem in quantum information theory is
    to determine the ability of a noisy channel to convey
    quantum information at a given standard of fidelity.
    The \emph{quantum capacity} measures 
    the optimal asymptotic rate of transmission
    (in qubits per channel use)
    possible for arbitrarily good fidelities
    (if not \emph{perfect} fidelity).
    The LSD (Lloyd~\cite{1997-Lloyd},
    Shor~\cite{2002-Shor},
    Devetak~\cite{2005-Devetak}) Theorem shows that the
    quantum capacity is equal to the regularised coherent information,
    an optimization that involves unlimited number of copies of the channel.
    Our understanding of the quantum capacity remains 
    limited -- given a simple memoryless channel (such as the 
    qubit depolarizing channel for certain error parameter), 
    determining whether it has a positive quantum capacity
    is not known to be decidable. 
    To gain insights into the often intractable problem of
    determining quantum capacities of channels, 
    ``assisted capacities'' have been studied (see e.g.
    \cite{2006-BennettDevetakShorSmolin}), where 
    the sender and the receiver are given extra free resources, 
    such as entanglement or classical communication.  

    In this paper we are interested in the \emph{non-asymptotic}
    (or finite blocklength) regime 
    focusing on the trade-off between the dimension of the quantum
    system to be sent, the number of channel uses made, and the
    fidelity achieved.  In the absence of
    feedback in the coding protocol, this is also
    called the `one-shot' regime since we can treat multiple
    channel uses as a single use of a larger channel.
    In the one-shot regime, we can remove assumptions such 
    as memoryless channel uses, address questions concerning 
    quantum error correcting codes, and understand how fast 
    the achievable rate converges to the capacity as 
    the number of uses increases.  Sometimes, one-shot studies 
    provide results concerning asymptotic capacities.  
    However, the exact trade-off of interest is generally intractable.  
    Even in the classical case, it is not practical
    to compute the obtainable region of parameters exactly,
    but quite powerful bounds are known \cite{PPV}.
    Parallel to the study of assisted capacities, one can consider 
    assisted codes in the finite blocklength regime.
    
    Mosonyi and Datta \cite{MD}, Wang and Renner \cite{RennerWang} and
    Renes and Renner \cite{RR} have given one-shot converse and
    achievability bounds for classical data transmission by unassisted
    codes over classical-quantum channels.  In \cite{DH} Datta and
    Hsieh derive converse and achievability results for classical and
    quantum data transmission by entanglement-assisted codes over
    general quantum channels in terms of smoothed min- and
    max-entropies.  
A drawback of the bounds given in \cite{DH}
is that no explicit method of computation is given,
and it is not clear that an efficient method exists.
A one-shot converse bound for entanglement-assisted codes
amenable to computation was given in Matthews and Wehner
\cite{1210.4722} by generalising the
hypothesis-testing based `meta-converse' of \cite{PPV} to quantum
channels. In particular, the bound is a semidefinite program (SDP).

An alternative approach to upper bound one-shot performance is to
optimize data transmission over a larger class of coding procedures
which is mathematically easier to describe.
This type of approach is applied to the related task of
entanglement distillation in an early paper by Rains \cite{2001-Rains-SDP},
which gives one-shot converse bounds for entanglement distillation
by local operations and classical communication in the form
of an SDP for the performance of the more
powerful class of PPT-preserving operations, along with 
many other insightful results.
This was also the approach used in \cite{2012-Matthews},
which derives a linear program for the performance
of transmitting classical data via classical channels by
codes which are \emph{non-signalling} when the encoder and
decoder are considered as a single bipartite operation.
The linear program was shown to be equivalent to the meta-converse of
of \cite{PPV}.  
Our paper follows this approach.  We consider quantum data
transmission via quantum channels using codes that are
non-signalling, PPT-preserving or both, when viewed as bipartite
operations.  We derive one-shot
correspondences that allow our results to be viewed as
extensions to results in \cite{1210.4722} and \cite{2001-Rains-SDP}.

    The structure of the paper along with a summary of our results are 
    as follows.  

    We start with some 
    mathematical and notational preliminaries in Section 
    \ref{prelims}.  
    Generally speaking, a ``code'' refers to a set of operations
    performed by the sender Alice and the receiver Bob that, when
    combined with the given channel uses, effects the data 
    transmission.
    In Section \ref{code-classes}
    we define a very general class of codes, 
    the \emph{forward-assisted codes}, which 
    can be implemented by local operations and
    forward (i.e. Alice to Bob)
    quantum communication over an arbitrary auxiliary channel
    (in addition to the use of the given noisy channel).
    This class includes a number of important,
    operationally defined subclasses:
    the \emph{unassisted codes},
    which only use local operations;
    the \emph{entanglement-assisted codes},
    where the auxiliary channel is only used to share
    entanglement between Alice and Bob before
    the local operations are applied;
    and the \emph{forward-classical-assisted codes},
    where the auxiliary channel is classical.  
    We use the fact that forward-assisted codes correspond to
    \emph{bipartite operations} which are non-signalling
    from Bob to Alice to define subclasses of
    forward-assisted code based on constraints
    on these bipartite operations.  
    The \emph{non-signalling codes} are those where
    the bipartite operation is also non-signalling from Alice to Bob.
    This class includes unassisted and entanglement-assisted codes.
    The \emph{PPT-preserving codes} are those for which
    the bipartite operation is PPT-preserving.
    This class includes all unassisted and forward-classical-assisted codes, 
    but not all entanglement-assisted codes. 
    Section \ref{code-classes} provides precise definitions of all
these classes and describes the relationships between them.

    Section \ref{SDPs} contains our main technical contribution.  
    We derive simple semidefinite programs (SDPs)
    for the optimal \emph{channel fidelity} of codes which
    are non-signalling, PPT-preserving, or both.

    In section \ref{NScodes}, we present the first application of our
    SDPs.  We compare our optimal channel fidelity for non-signalling
    codes with an earlier upper bound for entanglement-assisted codes
    (derived with different techniques in \cite{1210.4722} for 
    the \emph{success probability} of
    classical data transmission).  Surprisingly, our new bound, which
    applies to a larger class of codes, is at least as tight as the old
bound. Furthermore, from the asymptotic analysis of the
    earlier bound \cite{1210.4722}, we obtain a new asymptotic result
    for memoryless noisy channels: that entanglement-assisted and
    non-signalling codes give the same capacity.

In section \ref{PPTpCD}, we study optimal channel fidelity for
PPT-preserving codes.  We derive connections between PPT-preserving
codes and PPT-preserving entanglement distillation scheme studied in
by Rains in \cite{2001-Rains-SDP}.  We show that Rains' SDP for the
fidelity of PPT-preserving entanglement distillation provides lower
bounds on the fidelity of the PPT-preserving codes. We also show that
for certain special channels Rains' SDP coincides with our SDP for the
fidelity of PPT-preserving codes.

    In section \ref{WH}, we apply our SDPs to a concrete example,
    computing the fidelity for codes that are
    PPT-preserving, non-signalling or both, over the
    Werner-Holevo channels for blocklengths up to 120.
    The results demonstrate that
    codes which satisfy both constraints 
    can be strictly less powerful
    than codes that satisfy either one of the constraints.
    Thus combining the PPT-preserving and non-signalling constraints 
    provides strictly stronger upper bounds for unassisted 
    communication, at least for
    finite block-lengths.
    The results suggest that this improvement may
    even persist in the asymptotic regime.
      
    Furthermore, the results of section \ref{PPTpCD} and Rains
    \cite{2001-Rains-SDP} imply that PPT-preserving codes enable
    zero-error quantum communication (of one qubit) over two uses of
    three-dimensional Werner-Holevo channel.
    Surprisingly, the same holds even if the codes are also 
    non-signalling.  
    We discuss the relationship of this phenomenon to
    the superactivation of quantum capacity \cite{2008-SmithYard}.
    Our result
    could be considered a form of superactivation, 
    since neither the channel nor the code involved has 
    quantum capacity, yet their combination can communicate quantum 
    data perfectly.  
    However, we do not know whether the code can
    be implemented by local operations and forward communication
    over a channel with no quantum capacity. If it could be,
    then our result would demonstrate a very strong version
    of superactivation in the sense of \cite{2008-SmithYard},
    where two \emph{channels} with no quantum capacity
    could be used together to
    transmit quantum information \emph{perfectly}.
    In this connection, we show, via an example, that not all 
    PPT-preserving and non-signalling codes can be simulated by zero 
    capacity forward quantum channel.  

%SUMMARY OF RESULTS
%
%STRUCTURE OF PAPER
    \section{Preliminaries}\label{prelims}

    In this section, we summarize mathematical concepts required for
    the results.  We will also define unambiguous conventions 
    concerning our notation for quantum states and operations,
    which help us avoid a proliferation of brackets and
    tensor product symbols.

    A quantum system $\Q$ is associated to
    a Hilbert space $\mc{H}_{\Q}$ of dimension
    $\dim (\Q)$ (in this work we only deal with
    finite dimensional systems)
    and is equipped with a real, orthonormal `computational basis'
    $\{ \ket{i}_{\Q}: i = 1, \ldots, d \}$.
    We will always write linear operators on $\mc{H}_{\Q}$
    with a subscript identifying the system they act on, 
    for example, $X_{\Q}$.

    We assume that there is some fixed underlying order on systems
    which determines the order in which tensor products are taken.  
    We can write a product of operators acting on disjoint subsystems
    without the $\ox$ symbol, by taking it as given
    that the operators are padded with appropriate
    identity operators. For example,
    $X_{\Q}Y_{\R} = Y_{\R}X_{\Q} = X_{\Q}\ox Y_{\R} = (X_{\Q} \ox
    \1_{\R})(\1_{\Q}\ox Y_{\R})$. The same applies to a product of
    operators acting on different but not necessarily disjoint
    subsystems, for example, $X_{\P\Q}Y_{\Q\R} = (X_{\P\Q}\ox\1_{\R})
    (\1_{\P} \ox Y_{\Q\R})$.
 
    An \emph{operation} $\mc{N}_{\R\<\Q}$ (or \emph{channel})
    with input system $\Q$
    and output system $\R$
    is a completely positive, trace preserving linear map
    from the bounded linear operators on $\mc{H}_{\Q}$ to the bounded 
    linear operators on $\mc{H}_{\R}$.  Since we only deal with finite
    dimensional systems, all linear operators are bounded.
    As with operators, we always explicitly write
    the input and the output systems as subscripts.
    We write the set of all such operations as
    $\ops(\Q \to \R)$.  
    Our subscript convention has one exception:
    the trace operation on $\Q$, $\tr_{\Q}$, has the trivial,
    one-dimensional, output system, so we only write the input system.
    
    We denote the \emph{transpose map} on system $\Q$
    by $\t_{\Q\<\Q}$.
    It is the trace preserving, but not completely positive,
    linear map such that
    $\t_{\Q\<\Q}: \ketbra{i}{j}_{\Q} \mapsto \ketbra{j}{i}_{\Q}$.
    We also make use of the conventional notation
    $X^{\T}_{\Q}$ for $\t_{\Q\<\Q} X_{\Q}$.  
    
    Given two systems $\Q$ and $\tilde{\Q}$ of equal dimension,
    we can identify states of $\Q$ with states of $\tilde{\Q}$
    via the \emph{identity operation}
    $\id_{\tilde{\Q}\<\Q}: \ketbra{i}{j}_{\Q} \mapsto
    \ketbra{i}{j}_{\tilde{\Q}}$.
    Furthermore, we denote the isotropic maximally entangled
    state of $\tilde{\Q}\Q$ by
    $\phi_{\tilde{\Q}\Q} :=  \ketbra{\phi}{\phi}_{\tilde{\Q}\Q}$,
    \begin{equation}
    \ket{\phi}_{\tilde{\Q}\Q} :=
    \dim(\Q)^{{-}1/2} \sum_{i=1}^{\dim{\Q}}
    \ket{i}_{\tilde{\Q}}
    \ket{i}_{\Q} \,.
    \label{mes}
    \end{equation}
    A useful fact, sometimes called the `transpose trick',
    is that for any operator $M_{\Q}$ on $\Hs_{\Q}$, we have
    \begin{equation}\label{ttrick}
       M_{\Q} \ket{\phi}_{\tilde{\Q}\Q} =
       M_{\tilde{\Q}}^{\T} \ket{\phi}_{\tilde{\Q}\Q},
    \end{equation}
    where $M_{\tilde{\Q}} := \id_{\tilde{\Q}\<\Q} M_{\Q}$.

    To denote the application of a linear map $\mc{N}_{\R\<\Q}$
    to an operator $X_{\Q}$, we write simply $\mc{N}_{\R\<\Q}X_{\Q}$,
    just as we would write the application
    of a matrix to a vector without parenthesis.
    Products of operations represent compositions, with
    a convention similar to that defined for operators above,  
    so that tensor symbols and identity operations are omitted.  
    For example,
    $\mc{N}_{\R\<\Q} \mc{M}_{\T\<\P} X_{\Q\P} =
    (\mc{N}_{\R\<\Q} \ox \mc{M}_{\T\<\P}) X_{\Q\P}$, and 
    $\mc{N}_{\R\<\Q} X_{\Q\P} = (\mc{N}_{\R\<\Q} \ox \id_{\P\<\P}) X_{\Q\P}$. 

    We adopt the convention that
    multiplication of operators takes precedence over the
    application of linear maps from operators to operators,
    such as operations or the transpose map.
    For example $\t_{\Q\<\Q} X_{\Q} Y_{\Q} = \t_{\Q\<\Q} (X_{\Q} Y_{\Q})
    $, and 
    $\tr_{\Q} X_{\P\Q} Y_{\Q\R} = \tr_{\Q} (X_{\P\Q} Y_{\Q\R})$.
    
    To further illustrate these notational conventions,
    we note a useful fact 
    \begin{eqnarray}
    \hspace*{-3ex} & \tr_{\Q} X_{\P\Q} \t_{\Q\<\Q} Y_{\Q} 
    = \tr_{\Q} ( X_{\P\Q}  (\1_{\P} \otimes (\t_{\Q\<\Q} Y_{\Q})) ) 
    \nonumber
    \\ \hspace*{-3ex} & = \tr_{\Q} ( (\t_{\Q\<\Q} X_{\P\Q})  
    (\1_{\P} \! \otimes \! Y_{\Q}) ) 
    = \tr_{\Q} (\t_{\Q\<\Q} X_{\P\Q}) Y_{\Q} .
    \label{usefulfact}
    \end{eqnarray}
   
    In this paper, we define the \emph{Choi matrix} $N_{\R\Q}$ of an 
    operation $\mc{N}_{\R\<\Q}$
    to be the unique operator on $\mc{H}_{\R}\ox\mc{H}_{\Q}$
    such that for all operators $X_{\Q}$ on $\mc{H}_{\Q}$, 
\begin{equation}
\mc{N}_{\R\<\Q} X_{\Q} = \tr_{\Q} N_{\R\Q} \t_{\Q\<\Q} X_{\Q} 
= \tr_{\Q} (\t_{\Q\<\Q} N_{\R\Q}) X_{\Q} 
\label{choi2channel}
\end{equation} 
where the last equality comes from Eq.\ (\ref{usefulfact}).
    Our Choi matrix is equal to the common definition:  
    \begin{equation}
        N_{\R\Q} = \dim(\Q) \;
        \id_{\Q\<\tilde{\Q}} \mc{N}_{\R\<\Q} \phi_{\tilde{\Q}\Q}.
        \label{choi}
    \end{equation}
    We adopt the convention that where operations
    are denoted by a calligraphic letter,
    the corresponding Choi matrix is the same letter in the regular font.
    
A bipartite operator $X_{\P\Q}$ 
is said to be PPT (positive partial-transpose) if
$\t_{\P\<\P} X_{\P\Q} \geq 0$.
This condition is equivalent to $\t_{\Q\<\Q} X_{\P\Q} \geq 0$,
and is independent of the basis in which the transpose is taken.

An operation $\mc{F}_{\bo\<\ai}$ is called 
a `Horodecki' channel (or PPT-binding channel) if its Choi matrix 
$F_{\bo\ai}$ is PPT \cite{2000-Horodecki3}.

Let $\tilde{\ai}$ and $\tilde{\bi}$ be arbitrary systems in
the possession of Alice and Bob, respectively.
A bipartite operation
$\mc{Z}_{\ao\bo\<\ai\bi}$
is `PPT-preserving' \cite{1999-Rains,2001-Rains-SDP}
if it takes any state
which is PPT with respect to the
Alice / Bob partition to another PPT state.  
In other words, 
$\t_{\bi\tilde{\bi}\<\bi\tilde{\bi}}
\rho_{\ai\tilde{\ai}\bi\tilde{\bi}} \geq 0$
implies $\t_{\bo\tilde{\bi}\<\bo\tilde{\bi}} \mc{Z}_{\ao\bo\<\ai\bi} 
\rho_{\ai\tilde{\ai}\bi\tilde{\bi}} \geq 0$.  
As shown in \cite{2001-Rains-SDP}, a bipartite operation
$\mc{Z}_{\ao\bo\<\ai\bi}$ is PPT-preserving if and only if
its Choi matrix $Z_{\ao\bo\ai\bi}$ is PPT, that is
\begin{equation}\label{ChoiPPT}
    \t_{\bi\bo\<\bi\bo} Z_{\ao\bo\ai\bi} \geq 0.
\end{equation}

The PPT-preserving operations include all operations that can be
implemented by local operations and arbitrary rounds of two-way
classical communication (these are known as `LOCC' operations).
In fact, the PPT-preserving operations include even those 
implemented by local operations and arbitrary rounds of two-way
communication over Horodecki channels.
To see this, note that a Horodecki channel $\mc{F}_{\ai\<\bo}$ is a
degenerate PPT-preserving bipartite operation where
$\dim \ao = \dim \bi = 1$,
and the class of PPT-preserving operations is closed under composition.  

A bipartite operation
$\mc{Z}_{\ao\bo\<\ai\bi}$
is non-signalling from Bob to Alice if
$\tr_{\bo} \mc{Z}_{\ao\bo\<\ai\bi} = \mc{Z}^{\Alice}_{\ao\<\ai} \tr_{\bi}$
for some operation $\mc{Z}^{\Alice}_{\ao\<\ai}$.
That is, the marginal state of Alice's output is
given by some fixed operation applied to the marginal
state of Alice's input.
The equivalent condition on the Choi matrix
$Z_{\ao\bo\ai\bi}$ is
\begin{equation}\label{NSBA}
    \tr_{\bo} Z_{\ao\bo\ai\bi} = Z^{\Alice}_{\ao\ai}\1_{\bi},
\end{equation} where $Z^{\Alice}_{\ao\ai}$ is the Choi
matrix for $\mc{Z}^{\Alice}_{\ao\<\ai}$.
As a Choi matrix, $Z^{\Alice}_{\ao\ai}$ must
satisfy $\tr_{\ao} Z^{\Alice}_{\ao\ai} = \1_{\ai}$,
so (\ref{NSBA}) implies that
$Z^{\Alice}_{\ao\ai} = \tr_{\bo\bi} Z_{\ao\bo\ai\bi} / \dim(\bi)$.
Similarly, $\mc{Z}_{\ao\bo\<\ai\bi}$
is non-signalling from Alice to Bob if
\begin{equation}\label{NSAB}
    \tr_{\ao} Z_{\ao\bo\ai\bi} = Z^{\Bob}_{\bo\bi}\1_{\ai},
\end{equation}
where
$Z^{\Bob}_{\bo\bi} = \tr_{\ao\ai} Z_{\ao\bo\ai\bi} / \dim(\ai)$.
These conditions are quantum generalizations
of the classical non-signalling conditions
on bipartite conditional probability distributions.
One-way non-signalling operations have also
been referred to as `semi-causal' in the literature \cite{2001-BeckmanGottesmanNielsenPreskill,2002-EggelingSchlingemannWerner}.

\section{Classes of quantum codes}\label{code-classes}

In this section we define a very general class of codes, the 
\emph{forward-assisted codes}, and then various code subclasses
with operational or mathematical significance.

We represent the use of the noisy channel connecting Alice
to Bob by an operation $\mc{N}_{\bi\<\ao}$.
A \emph{forward-assisted code} is one which has the form
illustrated in Figure \ref{codefig}. The state to be transmitted
by Alice resides on a system $\ai$ with $\dim(\ai) = K$.
Alice performs an encoding map
$\mc{E}_{\ao\Q\<\ai}$ and sends the output systems through the 
noisy channel $\mc{N}_{\bi\<\ao}$ and some arbitrary side
channel $\mc{F}_{\R\<\Q}$. Then Bob applies
a local decoding operation $\mc{D}_{\bo\<\R\bi}$, where
the system $\bo$ has $\dim(\bo) = K$.
This results in an overall operation
$\mc{M}_{\bo\<\ai} = 
\mc{D}_{\bo\<\R\bi}\mc{F}_{\R\<\Q}\mc{N}_{\bi\<\ao}\mc{E}_{\ao\Q\<\ai}
\in \ops(\ai \to \bo)$.
We call the dimension $K$ the \emph{size} of the code.

We note that codes for multiple channel uses which
make use of some form of \emph{feedback} between the uses
(for example, codes assisted by two-way classical communication)
do not necessarily fall into the class of forward-assisted
codes.

\def\F{F} %% channel fidelity symbol
\newcommand{\mF}[3]{F^{#1}(#2,#3)} %% max fidelity of code
Given two systems $\tilde{\Q}$ and $\Q$ of equal dimension,
the \emph{entanglement fidelity of a state} $\sigma_{\tilde{\Q}\Q}$
is $\tr_{\tilde{\Q}\Q} \phi_{\tilde{\Q}\Q} \sigma_{\tilde{\Q}\Q}$.
Given $\mc{M}_{\bo\<\ai} \in \ops(\ai \to \bo)$ with
$\dim \ai = \dim \bo$, we follow \cite{2003-KretschmannWerner}
in calling
\[
    \F(\mc{M}_{\bo\<\ai}) =
    \tr_{\bo\tilde{\ai}}
    \phi_{\bo\tilde{\ai}} \mc{M}_{\bo\<\ai}\phi_{\ai\tilde{\ai}}
\]
the \emph{channel fidelity} of $\mc{M}_{\bo\<\ai}$.
When Alice's input is half of a maximally entangled state
$\phi_{\ai\tilde{\ai}}$
the overall effect of the encoded transmission yields a state
$\tau_{\bo\tilde{\ai}}$, as shown in the figure.
The channel fidelity of $\mc{M}_{\bo\<\ai}$ is the
entanglement fidelity of $\tau_{\bo\tilde{\ai}}$,
and we call this the \emph{channel fidelity of the code}.

The encoding procedure results in some average
\emph{channel input state}, which we will denote by
$\rho_{\ai} := \tr_{\Q\aic} \mc{E}_{\ao\Q\<\ai} \phi_{\ai\aic}$
(also shown in the figure).

\begin{figure}[h]
    \centering
    \input{fig1.tex}
    \caption{A forward-assisted-code is used to transmit half
    of a maximally entangled state $\phi_{\ai\tilde{\ai}}$
    over a noisy channel $\mc{N}_{\bi\<\ao}$.
    We can regard the forward-assisted-code as a 
    deterministic supermap,
    taking $\mc{N}_{\bi\<\ao}$ to the operation
    $\mc{M}_{\bo\<\ai}$ (with the dotted outline),
    which acts on $\phi_{\ai\tilde{\ai}}$.
    This supermap is determined by the bipartite operation
    $\mc{Z}_{\ao\bo\<\ai\bi}$ with the dashed outline.}
    \label{codefig}
\end{figure}
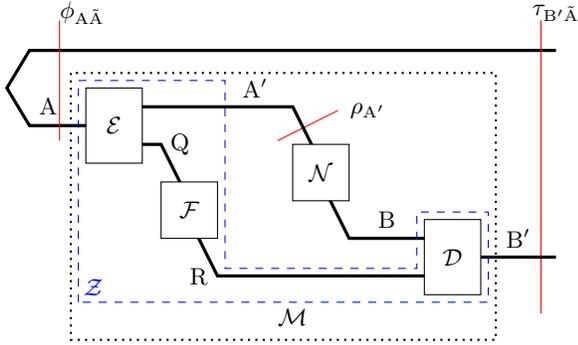

Consider the bipartite operation
\begin{equation}
\mc{Z}_{\ao\bo\<\ai\bi}
:= \mc{D}_{\bo\<\R\bi}\mc{F}_{\R\<\Q}\mc{E}_{\ao\Q\<\ai}\,,
\label{zop}
\end{equation}
which is outlined with dashes in Figure \ref{codefig}.
Using (\ref{choi2channel}),
its Choi matrix $Z_{\ao\bo\ai\bi}$ satisfies
\begin{equation}\label{Zchoi-eq}
    Z_{\ao\bo\ai\bi} =
    \tr_{\q\r} D_{\bo\bi\r} \t_{\r\<\r} F_{\r\q} \t_{\q\<\q}E_{\q\ao\ai}.
\end{equation}
Since this operation is implemented by local operations
and one-way quantum communication from Alice to Bob\footnote{
Such an operation is called ``semilocalisable'' in \cite{2001-BeckmanGottesmanNielsenPreskill}
},
it is non-signalling from Bob to Alice\footnote{Such an operation is called ``semicausal'' in
\cite{2001-BeckmanGottesmanNielsenPreskill}} \cite{2001-BeckmanGottesmanNielsenPreskill}
.
Conversely, \cite{2002-EggelingSchlingemannWerner} shows 
that any bipartite operation which is non-signalling
from Bob to Alice has an implementation by local operations
and one-way quantum communication from Alice to Bob.

In \cite{2008-ChiribellaDArianoPerinotti}, a {\em deterministic supermap}
$\mathfrak{M}$ is defined as a linear map from operations to
operations, such that tensoring $\mathfrak{M}$ with the identity
supermap still takes operations to operations.
In this language, the forward-assisted code
depicted in Figure \ref{codefig} constitutes 
a supermap from $\ops(\ao \to \bi)$ into $\ops(\ai \to \bo)$
\begin{equation}\label{op-map}
    \hspace*{-1.5ex} \mc{N}_{\bi\<\ao} \mapsto \mc{M}_{\bo\<\ai} 
  \; {=}\; \mc{D}_{\bo\<\R\bi}\mc{F}_{\R\<\Q}\mc{N}_{\bi\<\ao}\mc{E}_{\ao\Q\<\ai} \,.
\end{equation}
In \cite{2008-ChiribellaDArianoPerinotti}, it is shown that
any deterministic supermap from $\ops(\ao \to \bi)$
to $\ops(\ai \to \bo)$ can be implemented as in Figure \ref{codefig} and
eq.~(\ref{op-map}).
By expressing the Choi matrix $M_{\bo\ai}$
in terms of the Choi matrices of constituent operations using 
Eqs.~(\ref{choi2channel})-(\ref{choi}) and then using
Eq.~(\ref{Zchoi-eq}), one finds that
\[
    M_{\bo\ai} = \tr_{\ao\bi} Z_{\ao\bo\ai\bi} N_{\bi\ao}^{T}.
\] 
Therefore, the
action of a forward-assisted code, as a deterministic supermap,
is completely determined by the corresponding bipartite operation. In particular,
its channel fidelity is
\begin{equation}\label{EF-choi}
    K^{-1} \tr \phi_{\bo\ai} M_{\bo\ai} 
    = K^{-1} \tr \phi_{\bo\ai}
    Z_{\ao\bo\ai\bi} N_{\bi\ao}^{T}
\end{equation}
and its channel input state is
\begin{equation}\label{rho-choi}
    \rho_{\ao} =
    \tr_{\ai\bi\bo}
    Z_{\ao\bo\ai\bi} \1_{\ai} \1_{\bi}/\dim(\ai)\dim(\bi).
\end{equation}

Thus the set of forward-assisted codes of size $K$ for the channel
use $\mc{N}_{\bi\<\ao}$ corresponds precisely to
the set of deterministic supermaps from
$\ops(\ao \to \bi)$ to $\ops(\ai \to \bo)$, where $\dim(\ai) = \dim(\bo) = K$,
and thus to the set of 
bipartite operations $\ops(\ai:\bi \to \ao :\bo)$ which are non-signalling
from Bob to Alice. 

While the preceding discussion shows that the class of forward
assisted codes is mathematically natural to define, 
the class is too powerful to be interesting -- perfect performance is
trivially achieved for any $K$ and $\mc{N}_{\bi\<\ao}$, by 
choosing $\mc{F}_{\r\<\q}$ to be a $K$ dimensional quantum identity 
channel and by using $\mc{F}_{\r\<\q}$ to  transmit $\ai$ to Bob without even
using $\mc{N}_{\bi\<\ao}$. We now define several 
more interesting subclasses of the forward-assisted codes,
whose relationships are depicted in Figure \ref{codeclasses}.

The first three classes are operationally motivated -
that is they place further constraints on the way in which
the code can be implemented.
A conventional, unassisted quantum error correcting code
corresponds to not allowing \emph{any} forward assistance.
Equivalently, the operation $\mc{Z}_{\ao\bo\<\ai\bi}$ must have
the product form
$\mc{Z}_{\ao\bo\<\ai\bi} = \mc{D}_{\bo\<\bi}\mc{E}_{\ao\<\ai}$.
The operations $\mc{D}_{\bo\<\bi}$ and $\mc{E}_{\ao\<\ai}$ are still
arbitrary. We call this subclass \emph{unassisted codes} ($\UA$).
The strictly larger class of \emph{entanglement-assisted codes} ($\EA$)
corresponds to bipartite operations of the form
$\mc{Z}_{\ao\bo\<\ai\bi} =
\mc{D}_{\bo\<\bi\b}\mc{E}_{\ao\<\ai\a}\psi_{\a\b}$, where $\psi_{\a\b}$
can be any shared entangled state of arbitrary systems $\a$ and $\b$.
The class of forward-classical-assisted codes $\FCA$, is the
subclass of forward-assisted codes where we demand that
the auxiliary channel $\mc{F}_{\R\<\Q}$ is \emph{classical}. This means
that $\mc{F}_{\R\<\Q} \cdp_{\Q\<\Q} = \mc{F}_{\R\<\Q}$
and $\cdp_{\R\<\R} \mc{F}_{\R\<\Q}= \mc{F}_{\R\<\Q}$, where
$\cdp_{\Q\<\Q}$ denotes the completely dephasing operation in
the classical basis on $\Q$.

While the unassisted codes, the entanglement-assisted codes, and the
forward-classical-assisted codes possess clear operational
interpretations, they are generally difficult to optimise over.
Related classes that are more tractable to optimise are often 
studied instead.   

For both entanglement-assisted codes and unassisted codes,
the operation $\mc{Z}_{\ao\bo\<\ai\bi}$ is not only non-signalling
from Bob to Alice, but also from Alice to Bob. We call the subclass
of forward-assisted codes which is non-signalling from 
Alice to Bob the \emph{non-signalling codes} ($\NS$).
The transmission of classical data using classical channels
by non-signalling codes was first studied in \cite{CLMW2}.
In \cite{2012-Matthews}, the performance of non-signalling 
codes is used to provide a computationally tractable upper bound 
on unassisted classical codes over classical channels.   
The upper bound is equivalent to a powerful bound obtained using 
different methods in \cite{PPV}.

Unassisted codes and forward-classical-assisted codes satisfy 
a tractable constraint that $\mc{Z}_{\ao\bo\<\ai\bi}$ is 
PPT-preserving. 
We denote the subclass of forward-assisted codes that are
\emph{PPT-preserving} ``$\PPTp$''.
$\PPTp$ also contains forward-Horodecki-assisted
codes $\FHA$, consisting of 
forward-assisted codes where $\mc{F}_{\R\<\Q}$
is a Horodecki channel.
Since classical channels are Horodecki,
the class $\FHA$ contains $\FCA$.
We note that entanglement assisted codes are generally 
not PPT-preserving. The relationships between
the various classes of codes described
above are summarised in Figure \ref{codeclasses}.

\begin{definition}
    Let $\mF{\Omega}{\mc{N}}{K}$ denote the maximum channel fidelity
    $\F(\mc{M}_{\bo\<\ai})$ of operations
    $\mc{M}_{\bo\<\ai} \in \ops(\ai\to\bo)$
    with $\dim \ai = \dim \bo = K$
    which can be obtained by applying
    a forward-assisted code in class $\Omega$ to $\mc{N}_{\bi\<\ao}$.
\end{definition}

We can now define, for any class of codes $\Omega$,
the asymptotic \emph{quantum capacity} $Q^{\Omega}(\mc{N})$
of the \emph{memoryless} channel whose operation for
$n$ channel uses is $\mc{N}^{\ox n}$:
\begin{definition}
    \begin{equation}
        Q^{\Omega}(\mc{N})
        := \sup \{ r : \lim_{n \to \infty} F^{\Omega}(\mc{N}^{\ox n}, \lfloor 2^{rn} \rfloor ) = 1\}.
    \end{equation}
    We also define a corresponding \emph{zero-error} capacity by
    \begin{equation}
        Q_0^{\Omega}(\mc{N})
        := \sup_n \max \left\{ \frac{1}{n} \log_2 K_n
        : F^{\Omega}(\mc{N}^{\ox n}, K_n) = 1 \right\}.
    \end{equation}
\end{definition}
Given the results of \cite{2003-KretschmannWerner},
$Q^{\UA}(\mc{N})$ is equivalent to other definitions of
the (unassisted) quantum capacity $Q(\mc{N})$ of $\mc{N}$.
No ``single-letter'' formula for this quantity is known.
The best general expression we have for it is the regularised
coherent information formula of the
LSD Theorem~\cite{1997-Lloyd,2002-Shor,2005-Devetak}. 
$Q^{\EA}(\mc{N})$ is the entanglement-assisted capacity
of $\mc{N}$ for which we have the single-letter formula
of Bennett, Shor, Smolin and Thapliyal \cite{BSST}:
\begin{equation}\label{BSST}
    Q^{\EA}(\mc{N}_{\bi\<\ao})
    = \frac{1}{2} \max_{\rho_{\ao}}
    I(\mathrm{R}:\bi)_{\mc{N}_{\bi\<\ao} \rho_{\mathrm{R}\ao} }
\end{equation}
where $\rho_{\mathrm{R}\ao}$ is a purification of $\rho_{\ao}$
and $I(\mathrm{R}:\bi)_{\sigma_{\mathrm{R}\bi}} :=
S(\sigma_{\mathrm{R}}) + S(\sigma_{\bi}) - S(\sigma_{\mathrm{R}\bi})$,
where $S$ is the von Neumann entropy function.

The relationships between the classes of codes described in this
section imply the following inequalities:
\begin{align}
    &\mF{\UA}{\mc{N}}{K} \leq \mF{\EA}{\mc{N}}{K} \leq \mF{\NS}{\mc{N}}{K},\label{UAinEAinNS}\\%[2ex] <- FOR TWO COLUMN
    %%%%%%%%%%%%%% SPLIT VERSION FOR TWO COLUMN
    \begin{split}
    &~~\mF{\UA}{\mc{N}}{K} \leq \mF{\FCA}{\mc{N}}{K}\\
    &~~~~~~~~\leq \mF{\FHA}{\mc{N}}{K} \leq \mF{\PPTp}{\mc{N}}{K},
    \end{split}\\[2ex]
    %%%%%%%%%%%%%% SINGLE LINE VERSIO FOR ONE coLUMN
    %&\mF{\UA}{\mc{N}}{K} \leq \mF{\FCA}{\mc{N}}{K}
    %\leq \mF{\FHA}{\mc{N}}{K} \leq \mF{\PPTp}{\mc{N}}{K},\\
    %%%%%%%%%%%%%%%%%%%%%%%%%%%%%%%%%%%%%%%%%%%%%%%%%%
    %&\mF{\UA}{\mc{N}}{K} \leq \mF{\NS \cap \PPTp}{\mc{N}}{K}.
\end{align}
Similar inequalities hold for the corresponding assisted capacities.

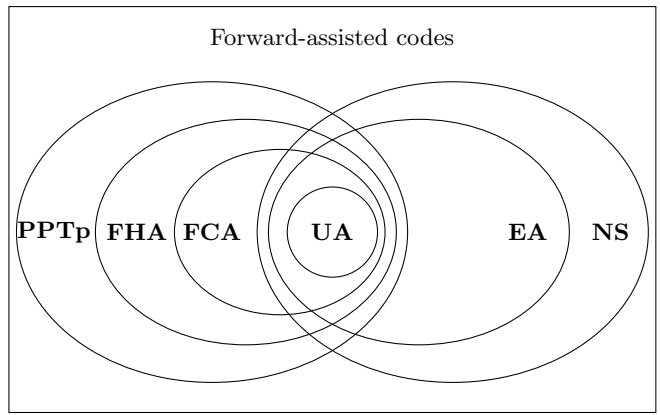
\begin{figure}
    \centering
    \begin{tikzpicture}
    \draw (-4.3,-2.4) rectangle (4.3,3);
    \node at (0,2.6) {Forward-assisted codes}; 

    \draw (0,0) ellipse (0.6cm and 0.6cm);   
    \node at (0,0) {$\UA$}; 
        
    %\draw (0.7,0) ellipse (2cm and 1.5cm);
    %\node at (2.15,0) {$\EA$}; 
    %\draw (1.4,0) ellipse (2.8cm and 2cm); 
    %\node at (3.25,0) {$\NS$}; 
    
    \draw (1.15,0) ellipse (2.0cm and 1.5cm);
    \node at (2.6,0) {$\EA$}; 
    \draw (1.6,0) ellipse (2.6cm and 2cm); 
    \node at (3.7,0) {$\NS$}; 
    
    \draw (-0.7,0) ellipse (1.4cm and 1.1cm);
    \node at (-1.6,0) {$\FCA$}; 
    \draw (-1.15,0) ellipse (2.0cm and 1.5cm);
    \node at (-2.6,0) {$\FHA$}; 
    \draw (-1.6,0) ellipse (2.6cm and 2cm); 
    \node at (-3.7,0) {$\PPTp$}; 
    \end{tikzpicture}
    \caption{
        The relationship between various subclasses
        of forward-assisted codes:
        PPT-preserving codes $\PPTp$;
        forward-Horodecki-assisted codes $\FHA$;
        forward-classical-assisted codes $\FCA$;
        unassisted codes $\UA$;
        entanglement-assisted codes $\EA$;
        non-signalling codes $\NS$;
    }\label{codeclasses}
\end{figure}

In the next section, we show how the optimal channel fidelity
of forward-assisted codes which are non-signalling, PPT-preserving,
or both can be formulated as semidefinite programs (SDPs) \cite{VB-SDP,Watrous-lecture-notes}. SDPs have a number of
attractive qualities:
there are efficient algorithms for performing the optimising
numerically; feasible points to the dual programs yield
upper bounds on the optimal performance;
in many cases of interest, \emph{strong duality}
holds, so that dual solutions can certify optimality.

\section{Semidefinite programs for PPT-preserving and
non-signalling codes}\label{SDPs}

We have seen that the full set of forward-assisted codes
of size $K$ for the channel operation $\mc{N}_{\bi\<\ao}$
corresponds to those bipartite operations in $\ops(\ai\bi \to \ao\bo)$
which are non-signalling from Bob to Alice, where
$\dim(\ai) = \dim(\bo) = K$.
The corresponding set of Choi matrices are those satisfying
\begin{align}
    Z_{\ao\bo\ai\bi} &\geq 0,\label{ChoiCP}\\
    \tr_{\ao\bo} Z_{\ao\bo\ai\bi} &= \1_{\ai\bi},\label{ChoiTP}\\
    \tr_{\bo} Z_{\ao\bo\ai\bi} &= \tr_{\bo\bi} Z_{\ao\bo\ai\bi} / \dim(\bi).
    \label{ChoiNSBA2}
\end{align}
Here ($\ref{ChoiCP}$), ($\ref{ChoiTP}$)
are equivalent  to the operation being completely positive and trace
preserving, respectively.
The equality ($\ref{ChoiNSBA2}$) is the constraint that the
operation is non-signalling from Bob to Alice (see ($\ref{NSBA}$)).

The code is non-signalling (see (\ref{NSAB})) if and only if
\begin{equation}\label{NSAB2}
    \NS: \tr_{\ao} Z_{\ao\bo\ai\bi} =
    \tr_{\ao\ai} Z_{\ao\bo\ai\bi} / \dim(\ai),
\end{equation}
and PPT-preserving (see (\ref{ChoiPPT})) if and only if
\begin{equation}\label{ChoiPPT2}
    \PPTp: \t_{\bi\bo\<\bi\bo} Z_{\ao\bo\ai\bi} \geq 0.
\end{equation}

As noted earlier (eqn.~(\ref{EF-choi})),
the channel fidelity is given by
\begin{equation}\label{EF-choi2}
    f_c = K^{-1} \tr \phi_{\bo\ai}
    Z_{\ao\bo\ai\bi} N_{\bi\ao}^{T}.
\end{equation}
The problem is to maximize $f_c$ subject to 
% the constraints 
(\ref{ChoiCP})-(\ref{ChoiNSBA2}),
with the additional constraints (\ref{NSAB2}), 
(\ref{ChoiPPT2}) as appropriate.

We begin by showing that we can, without loss of generality,
restrict our attention to a highly symmetric form of $Z_{\ao\bo\ai\bi}$.
Let $\bar{U}$ denote the complex conjugate of $U$,
and let $p$ denote the unique Haar probability measure on
the unitary group $U(K)$.
The channel fidelity  eq.~(\ref{EF-choi2}) satisfies
\begin{eqnarray}\label{EF-choi-twirled}
    & & K^{-1} \tr \phi_{\bo\ai}
    Z_{\ao\bo\ai\bi} N_{\bi\ao}^{T}
\nonumber
\\
    & = & K^{-1} \tr \int
         dp(U) U^\dagger_{\bo}U^T_{\ai} \phi_{\bo\ai} U_{\bo}
         \bar{U}_{\ai}
    Z_{\ao\bo\ai\bi} N_{\bi\ao}^{T}
\nonumber
\\
    & = & K^{-1} \tr 
         \phi_{\bo\ai}
         \bar{Z}_{\ao\bo\ai\bi}
         N_{\bi\ao}^{T},
\nonumber
\end{eqnarray}
where
\begin{equation}\label{twirled-Choi}
    \bar{Z}_{\ao\bo\ai\bi} := \int dp(U)
         U_{\bo} \bar{U}_{\ai}
         Z_{\ao\bo\ai\bi} 
         U^\dagger_{\bo}U^T_{\ai}.
\end{equation}
The first equality holds because  
    $U^\dagger_{\bo}U^T_{\ai}\ket{\phi}_{\bo\ai} = 
    \ket{\phi}_{\bo\ai}$ for all unitary operators $U$,
    by the `transpose trick' (Eq.\ (\ref{ttrick})).
The second equality follows from
the cyclic property and linearity of the trace.
If we define the `twirling' operation
\begin{equation}\label{twirling-op}
    \mc{T}_{\bo\ai\<\bo\ai}: X_{\bo\ai} \mapsto
    \int dp(U)
         U_{\bo} \bar{U}_{\ai}
         X_{\bo\ai} 
         U^\dagger_{\bo}U^T_{\ai},
\end{equation}
then $\bar{Z}_{\ao\bo\ai\bi} =
\id_{\bi\ao\<\bi\ao} \mc{T}_{\bo\ai\<\bo\ai}
Z_{\ao\bo\ai\bi}$.

Consider a general Choi matrix $N_{\R\Q}$ given by Eq.\ (\ref{choi}).
By the transpose trick, $W_\Q N_{\R\Q} W^\dagger_\Q$ is the Choi
matrix of the map that conjugates the input by $W^T$ before 
$\mc{N}_{\R\<\Q}$ acts.  Meanwhile, $W_\R N_{\R\Q} W^\dagger_\R$
is the Choi matrix of the map that first applies $\mc{N}_{\R\<\Q}$ 
before conjugation by $W_\R$.  
Therefore, the `twirled' operator in (\ref{twirled-Choi}) 
corresponds to the modified bipartite operation
$\bar{\mc{Z}}_{\ao\bo\<\ai\bi}[\cdot]
= 
\int dp(U)
U_{\bo}
\mc{Z}_{\ao\bo\<\ai\bi} [ U^{\dagger}_{\ai} \cdot U_{\ai} ] 
U^{\dagger}_{\bo}$.

The operation $\bar{\mc{Z}}_{\ao\bo\<\ai\bi}$
can be implemented as follows: Alice and Bob
share a classical random variable identifying a unitary $U$ drawn
according to the Haar measure $p$. Alice applies $U_{\ai}^{\dagger}$
to her input system $\ai$. Alice and Bob then use the forward assisted
code corresponding to $\mc{Z}$. Finally,
Bob applies $U_{\bo}$, inverting Alice's operation on the input.
Since $\mc{Z}_{\ao\bo\<\ai\bi}$ can be transformed to 
$\bar{\mc{Z}}_{\ao\bo\<\ai\bi}$ 
using local operations and shared randomness,
$\bar{\mc{Z}}_{\ao\bo\<\ai\bi}$ will be non-signalling
from Alice to Bob if $\mc{Z}_{\ao\bo\<\ai\bi}$ is,
and will be PPT-preserving if $\mc{Z}_{\ao\bo\<\ai\bi}$ is.

Equation (\ref{EF-choi-twirled})
tells us that, for any given $\mc{N}_{\bi\<\ao}$,
using the $\bar{\mc{Z}}_{\ao\bo\<\ai\bi}$ will yield the same
channel fidelity as using $\mc{Z}_{\ao\bo\<\ai\bi}$.
Therefore, there is no loss of generality in assuming
that the Choi matrix lies in the image of
the operation $\id_{\bi\ao\<\bi\ao} \mc{T}_{\bo\ai\<\bo\ai}$.

As shown in Rains \cite{2001-Rains-SDP},
the action of $\mc{T}_{\bo\ai\<\bo\ai}$ can also be written
\begin{equation}
    \begin{split} %%%%%%%% TWO COLUMN VERSION
    &\mc{T}_{\bo\ai\<\bo\ai}: X_{\bo\ai} \mapsto
    \phi_{\bo\ai} \tr \phi_{\bo\ai} X_{\bo\ai}
    +\hspace*{5ex}\\
    &\hspace*{5ex}\frac{(\1_{\bo\ai} - \phi_{\bo\ai})}
    {\tr (\1_{\bo\ai} - \phi_{\bo\ai})}
    \tr (\1_{\bo\ai} - \phi_{\bo\ai}) X_{\bo\ai}.
    \end{split}
    %\mc{T}_{\bo\ai\<\bo\ai}: X_{\bo\ai} \mapsto
    %\phi_{\bo\ai} \tr \phi_{\bo\ai} X_{\bo\ai}
    %+
    %\frac{(\1_{\bo\ai} - \phi_{\bo\ai})}
    %{\tr (\1_{\bo\ai} - \phi_{\bo\ai})}
    %\tr (\1_{\bo\ai} - \phi_{\bo\ai}) X_{\bo\ai}.
\end{equation}
Thus, $\bar{Z}_{\ao\bo\ai\bi}$ lies in the image of
$\id_{\bi\ao\<\bi\ao} \mc{T}_{\bo\ai\<\bo\ai}$
if and only if
    \begin{equation}\label{U-invar-form}
    \bar{Z}_{\ao\bo\ai\bi} = K(\phi_{\bo\ai} \sdpop_{\ao\bi}
    + (\1-\phi)_{\bo\ai} \sdpopQ_{\ao\bi}) \,,
\end{equation}
for some operators $\sdpop_{\ao\bi}$ and $\sdpopQ_{\ao\bi}$.
When we write $\sdpop$, $\sdpopQ$ subscripted with only $\ao$ 
or $\bi$, we refer to the partial traces of the operators, 
for example, $\sdpop_{\ao} {:}{=} \tr_{\bi} \sdpop_{\ao\bi}$.  
From (\ref{rho-choi}), we see that the modified forward-assisted code
(\ref{U-invar-form}) has channel input state
\begin{eqnarray}
    \rho_{\ao} & =& (K\dim(\bi))^{-1} \tr_{\bo\ai\bi} \bar{Z}_{\ao\bo\ai\bi}
\nonumber
\\
   & =& (\sdpop_{\ao} + (K^2-1)\sdpopQ_{\ao}) \dim(\bi)^{-1} \,.
\label{ave}
\end{eqnarray}

Expressing the constraints on $\bar{Z}$
in terms of $\sdpop_{\ao\bi}$ and $\rho_{\ao}$ gives 
the following theorem and corollary.  
\begin{theorem}\label{main-T}
        There is a forward-assisted code
        (see Figure \ref{codefig}) of size $\dM$,
        average channel input $\rho_{\ao}$ and channel fidelity
        $f_c$ for $\mc{N}_{\bi\<\ao}$ which is
        PPT preserving and/or non-signalling from Alice to Bob
        if and only if there exists an operator
        $\sdpop_{\ao\bi}$ such that
        \begin{align}
            &f_c = \tr N_{\ao\bi}^{\T} \sdpop_{\ao\bi} \label{fidelity}\\
            &\sdpop_{\ao\bi} \leq \rho_{\ao}\1_{\bi}\label{Qpos} \\
            &\sdpop_{\ao\bi} \geq 0\label{Rpos} \\
            \NS:& \sdpop_{\bi} =
            \1_{\bi}/\dM^2\label{NSABchoi} \\
            \PPTp:&
            \begin{cases}
                \t_{\bi\<\bi}[\sdpop_{\ao\bi}] \geq -\rho_{\ao}\1_{\bi}/K,\\
                \t_{\bi\<\bi}[\sdpop_{\ao\bi}] \leq \rho_{\ao}\1_{\bi}/K.
            \end{cases}
            \label{PPT-constraints}
        \end{align}
\end{theorem}
\begin{corollary}\label{primal-SDP}
        To obtain $\mF{\NS}{\mc{N}}{K}$ or $\mF{\PPTp}{\mc{N}}{K}$
        we maximise the expression (\ref{fidelity})
        subject to either (\ref{NSABchoi}) or (\ref{PPT-constraints}),
        as appropriate, in addition to
        the constraints (\ref{Rpos}), (\ref{Qpos}),
        $\rho_{\ao} \geq 0$, and $\tr \rho_{\ao} = 1$.
        If we impose both (\ref{NSABchoi}) and (\ref{PPT-constraints})
        we obtain $\mF{\NS \cap \PPTp}{\mc{N}}{K}$
        In all three cases, the optimisation is a semidefinite program.
    \end{corollary}
\begin{proof}
We begin by deriving the expression for the channel fidelity 
(\ref{fidelity}).  
It follows by substituting (\ref{U-invar-form}) into (\ref{EF-choi2}) and 
using
$(\1{-}\phi)_{\bo\ai}\phi_{\bo\ai} = 0$
and $\phi_{\bo\ai} \phi_{\bo\ai} = \phi_{\bo\ai}$.

We next consider the constraints (\ref{ChoiCP})-(\ref{ChoiNSBA2}).
Using Eqs.\ (\ref{U-invar-form}) and (\ref{ave}), we see that
(\ref{ChoiNSBA2}) is
equivalent to
\begin{equation}\label{NSBAtwirled}
    \sdpop_{\ao\bi} + (K^2-1)\sdpopQ_{\ao\bi}
    = \rho_{\ao}\1_{\bi}.
\end{equation}
We will use this relation to eliminate
$\sdpopQ_{\ao\bi}$ in the other constraints.
Substituting (\ref{U-invar-form})
into the `trace preserving' constraint (\ref{ChoiTP}), we obtain
\begin{equation}\label{nlz}
    \sdpop_{\bi} + (K^2 - 1) \sdpopQ_{\bi} = \1_{\bi} \,. 
\end{equation}
Note that eq.~(\ref{nlz}) is already implied by (\ref{NSBAtwirled}).  

Since $(\1-\phi)_{\bo\ai}$ and $\phi_{\bo\ai}$ are
positive-semidefinite operators supported on orthogonal subspaces,
$\bar{Z}_{\ao\bo\ai\bi}$ in eq.~(\ref{U-invar-form}) satisfies
the complete positivity constraint (\ref{ChoiCP})
if and only if $\sdpop_{\ao\bi}
\geq 0$ and $\sdpopQ_{\ao\bi} \geq 0$.
The first of these is constraint $(\ref{Qpos})$,
and $(\ref{Rpos})$ is obtained by using
(\ref{NSBAtwirled}) to substitute for $\sdpopQ_{\ao\bi}$
in the latter.

Now, if we want our forward-assisted code to be
non-signalling from Alice to Bob (satisfying (\ref{NSAB2}))
then, by eqs.~(\ref{U-invar-form}) and (\ref{nlz}), 
this is equivalent to
\begin{equation} 
K (\phi_{\bo\ai} + \sdpop_{\bi} + (\1-\phi)_{\bo\ai} \sdpopQ_{\bi}) 
= K^{-1} \1_{\ai \bo \bi} \,. 
\end{equation}
Eliminating $\sdpopQ_{\bi}$ using (\ref{nlz}), the above
holds if and only if $\sdpop_{\bi} = \1_{\bi}/\dM^2$, which is
constraint (\ref{NSABchoi}) in our Theorem.

Finally, we can show that $(\ref{U-invar-form})$ is PPT-preserving
(constraint (\ref{ChoiPPT2}))
if and only if conditions ($\ref{PPT-constraints}$) hold, in a way
similar to Rains \cite{2001-Rains-SDP}.  To see this, apply 
$\t_{\bi\bo\<\bi\bo}$ to both sides of $(\ref{U-invar-form})$.
Using the fact that
$\t_{\bo\<\bo} \phi_{\bo\ai} = (\sym_{\bo\ai} - \asym_{\bo\ai})/K$ and
$\sym_{\bo\ai} + \asym_{\bo\ai} = \1_{\bo\ai}$, where
$\sym_{\bo\ai}$ and $\asym_{\bo\ai}$ are the projectors onto the
symmetric and the antisymmetric subspaces of
$\mc{H}_{\ai}\ox\mc{H}_{\bo}$ respectively, one obtains 

    \begin{equation}
        \begin{split}
           \t_{\bi\bo\<\bi\bo}
        [\bar{Z}_{\ao\bo\ai\bi}]
        & =
        % K^{-1}
        \sym_{\bo\ai}(\t_{\bi\<\bi}[\sdpop_{\ao\bi}
        + (K{-}1) \sdpopQ_{\ao\bi}]) 
\nonumber
        \\
        & + 
        % K^{-1}
        \asym_{\bo\ai}(\t_{\bi\<\bi}[-\sdpop_{\ao\bi}
        + (K{+}1) \sdpopQ_{\ao\bi}]) .
\nonumber
        \end{split}
    \end{equation}
    Using the fact that $\sym_{\bo\ai}$ and $\asym_{\bo\ai}$ are
    orthogonal projectors, this last expression
    is positive semidefinite if and only if 
$\t_{\bi\<\bi}[\sdpop_{\ao\bi}
        + (K{-}1) \sdpopQ_{\ao\bi}] \geq 0$ and 
$\t_{\bi\<\bi}[-\sdpop_{\ao\bi}
        + (K{+}1) \sdpopQ_{\ao\bi}] \geq 0$.  
Eliminating $\sdpopQ_{\ao\bi}$ using ($\ref{NSBAtwirled}$)
in these two conditions gives ($\ref{PPT-constraints}$).
\end{proof}

%%%%%%%%%%%%%%%%%%%%%%%%%%%%%%%%%%%%%%%%%%%%%%%%%%%%%%%%%%%%%%%%

We now derive the dual semidefinite program for the entanglement
fidelity achieved by a forward-assisted code that is
PPT-preserving and/or non-signalling,
using Lagrange multipliers.  The weak duality theorem states 
that the value of the dual program attained at any dual feasible
solution is at least the value of the primal program at any primal
feasible solution. Interested readers can consult
\cite{VB-SDP,Watrous-lecture-notes}.

%%%%%%%%%%%%%%%%%%%%%%%%%%%%%%%%%%%%%%%%%%%%%%%%%%%%%%%%%%%%%%%%
    \begin{proposition}
        The dual semidefinite program for
        $\mF{\NS \cap \PPTp}{\mc{N}}{K}$
        %)
        is to
        minimise $\mu + K^{-2} \tr W_{\bi}$ subject to
        \begin{align}
            N_{\ao\bi}^{T} + \t_{\bi\<\bi} \Omega_{\ao\bi}
            \leq X_{\ao\bi} + \1_{\ao}W_{\bi},
\label{dual1}           
\\
            \tr_{\bi} (X_{\ao\bi} + K^{-1}|\Omega_{\ao\bi}|)
            \leq \mu \1_{\ao},
\label{dual2}           
\\
            X_{\ao\bi} \geq 0 \,.
\label{dual3}           
        \end{align}
        To remove the PPT constraint, set $\Omega_{\ao\bi} = 0$.
        To remove the non-signalling constraint, set $W_{\bi} = 0$.
    \end{proposition}

    \begin{proof}
        We associate a positive-semidefinite Lagrange multiplier for
        each inequality constraint, and a hermitian Lagrange
        multiplier to each equality constraint.  In particular, we
        associate the operator $X_{\ao\bi} \geq 0$ to the constraint
        (\ref{Qpos}), a hermitian $W_{\bi}$ to non-signalling
        constraint (\ref{NSABchoi}), positive semidefinite 
        $Y_{\ao\bi}, V_{\ao\bi}$  
        to the PPT-preserving constraints
        (\ref{PPT-constraints}),  and a real multiplier $\mu$
        to the constraint that $\tr
        \rho_{\ao} = 1$.  The resulting Lagrangian is 
        \begin{align*}
            &\tr N^T_{\ao\bi} \sdpop_{\ao\bi} \\
        +& \tr X_{\ao\bi}(\rho_{\ao}\1_{\bi} - \sdpop_{\ao\bi})\\
        +& \tr Y_{\ao\bi}(\rho_{\ao}\1_{\bi}/K +
         \t_{\bi\<\bi} \sdpop_{\ao\bi} )\\
        +& \tr V_{\ao\bi}(\rho_{\ao}\1_{\bi}/K - 
        \t_{\bi\<\bi} \sdpop_{\ao\bi} )\\
        +& \tr \1_{\ao}W_{\bi} (\dim(\ao)^{-1} K^{-2} \1_{\ao\bi}
        - \sdpop_{\ao\bi})\\
        +& \mu(1-\tr \rho_{\ao})\\
        =& \tr \sdpop_{\ao\bi}(N^T_{\ao\bi} - X_{\ao\bi}
        + \t_{\bi\<\bi}[Y_{\ao\bi}{-}V_{\ao\bi}] 
        - \1_{\ao}W_{\bi})\\
        +& \tr \rho_{\ao}
        (\tr_{\bi}
        [X_{\ao\bi} + K^{-1}(Y_{\ao\bi} + V_{\ao\bi})]
        - \mu\1_{\ao})\\
        +& \mu + K^{-2} \tr W_{\bi}.
        \end{align*}
        The dual SDP is to minimise
        $\mu + K^{-2} \tr W_{\bi}$
        subject to
        \begin{align}
            N_{\ao\bi}^{T} + \t_{\bi\<\bi}[Y_{\ao\bi} - V_{\ao\bi}]
            \leq X_{\ao\bi} + \1_{\ao}W_{\bi},
\\
            \tr_{\bi} (X_{\ao\bi} + K^{-1}(Y_{\ao\bi} + V_{\ao\bi}))
            \leq \mu \1_{\ao},
\\
            X_{\ao\bi}, Y_{\ao\bi}, V_{\ao\bi} \geq 0 \,.
        \end{align}
        Let $\Omega_{\ao\bi} := Y_{\ao\bi}-V_{\ao\bi}$,
        then $|\Omega_{\ao\bi}| \leq Y_{\ao\bi} + V_{\ao\bi}$,
        and this can be made an equality by choosing
        $Y_{\ao\bi} = (|\Omega_{\ao\bi}| + \Omega_{\ao\bi})/2$
        and
        $V_{\ao\bi} = (|\Omega_{\ao\bi}| - \Omega_{\ao\bi})/2$,
        without loss of generality.

        Finally, to eliminate a constraint from the primal,
        we impose the additional constraint in the dual
        that the associated multiplier(s) be set to zero.
    \end{proof}
    
    An easy consequence of the dual for PPT-preserving codes
    is that their performance over Horodecki channels is no
    better than their performance over completely useless
    channels:
    %%%%%%%%%%%%%%%%%%%%%%%%%%%%%%%%%%%%%%%%%%%%%%%%%%%%%%%%%%%%%%%%
    \begin{proposition}
        The channel fidelity of a PPT-preserving 
        code for sending the state of a $K$-dimensional system over
        any Horodecki channel $\mc{N}^{H}$ is $1/K$ i.e.
        $\mF{\PPTp}{\mc{N}^{H}}{K} = 1/K$.
    \end{proposition}
    \begin{proof}
        First, the channel fidelity $1/K$ is achieved trivially
        without even using the Horodecki channel, by choosing
        $\mc{E}_{\ao\Q\<\ai}$ in Figure \ref{codefig} to be a
        measurement in the computational basis, $\Q$ to carry the
        measurement outcome, and $\mc{F}_{\R\<\Q}$ to be a noiseless
        classical channel of dimension $K$, and $\mc{D}_{\bo\<\R\bi}$ 
        to be the identity operation.  

        Second, to see $1/K$ is also an upper bound for the
        channel fidelity, we exhibit a dual feasible solution whose
        value in the dual SDP is $1/K$:
        Since we do not have the Alice to Bob non-signalling constraint,
        we must set $W_{\bi} = 0$.
        For this $W_{\bi}$, constraint (\ref{dual1}) is  
        implied by (\ref{dual3}) if we choose 
        $\Omega_{\ao\bi} = -\t_{\ao\<\ao} N_{\ao\bi}$.   
        Furthermore, since $\t_{\ao\<\ao} N_{\ao\bi} \geq 0$ for a
        Horodecki channel, $|\Omega_{\ao\bi}| = \t_{\ao\<\ao} N_{\ao\bi}$
        and $\tr_{\bi}|\Omega_{\ao\bi}| = \1_{\ao}$.
        Then, choosing $X_{\ao\bi} = 0$ and $\mu = 1/K$
        implies (\ref{dual2}) and (\ref{dual3}).  Together, 
        the above gives a dual feasible point with value $\mu = 1/K$.  
    \end{proof}

%%%%%%%%%%%%%%%%%%%%%%%%%%%%%%%%%%%%%%%%%%%%%%%%%%%%%%%%%%%%%%%%
        
        \section{Non-signalling codes}\label{NScodes}
        
        In this section, we compare the performance of
        entanglement-assisted codes and non-signalling codes.
        Furthermore, we show that the 
        entanglement-assisted classical
        capacity of any (memoryless) channel is equal to the
        non-signalling assisted classical capacity.

        First, recall from (\ref{UAinEAinNS}) that our SDP for
        non-signalling codes in Corollary \ref{primal-SDP} provides an
        upper bound on the fidelity of entanglement-assisted codes:
        \begin{equation}\label{NSBoundForEA}
            \mF{\EA}{\mc{N}}{K} \leq \mF{\NS}{\mc{N}}{K}.  
        \end{equation}

        \newcommand{\mPs}[3]{P_{s}^{#1}(#2,#3)}
        \def\Ps{P_{s}}
        Now, given free entanglement, there is a
        one-to-one correspondence between the performance for
        transmitting quantum and classical data. 

        The \emph{success probability} $\Ps(\mc{M}_{\bo\<\ai})$
        of an operation $\mc{M}$ is
        a measure of its ability to send classical data
        encoded in the computational basis:
        \[
            \Ps(\mc{M}_{\bo\<\ai}) =
            \frac{1}{K}\sum_{k=0}^{K-1}
            \tr_{\bo}
            \ketbra{k}{k}_{\bo}
            \mc{M}_{\bo\<\ai}\ketbra{k}{k}_{\ai}.
        \]
        \begin{definition}
            Let $\mPs{\Omega}{\mc{N}}{K}$
            denote the maximum success probability
            $\Ps(\mc{M}_{\bo\<\ai})$ of operations
            $\mc{M}_{\bo\<\ai} \in \ops(\ai\to\bo)$
            with $\dim \ai = \dim \bo = K$
            which can be obtained by applying
            a forward-assisted code in class $\Omega$ to $\mc{N}_{\bi\<\ao}$.
        \end{definition}

        The correspondence between the performance for transmitting
        quantum and classical data in the presence of free
        entanglement is due to superdense
        coding \cite{denseCoding} and teleportation \cite{teleportation}.  
        Using the superdense coding protocol,
        a $K$-dimensional quantum code of
        channel fidelity $f$ can be turned to a protocol for
        sending one out of $K^2$ equiprobable messages with success
        probability $f$.  The reverse holds by means of the teleportation 
        protocol. (See Appendix \ref{apdx} for details.)
        Therefore, for any subclass $\Omega$
        of forward-assisted codes that includes the entanglement-assisted
        codes (that is, $\EA \subseteq \Omega$) we have
        \begin{equation} \label{ea-correspondence}
            \mF{\Omega}{\mc{N}}{K} = \mPs{\Omega}{\mc{N}}{K^2}.
        \end{equation}
        For example, this equation holds for the class $\NS$
        of non-signalling codes.

        We can use the correspondence to obtain an upper bound on
        $\mF{\EA}{\mc{N}}{K}$ using the results in \cite{1210.4722}. 
        There, $\mPs{\EA}{\mc{N}}{K}$ is upper bounded by the solution 
        of a semidefinite program which we call $B(\mc{N},K)$.
        By (\ref{ea-correspondence}), $B(\mc{N},K)$ provides an SDP upper bound
        on $\mF{\EA}{\mc{N}}{K}$:
        \begin{equation}\label{SWBoundForEA}
            \mF{\EA}{\mc{N}}{K} = \mPs{\Omega}{\mc{N}}{K^2} \leq B(\mc{N},K^2).
        \end{equation}
        
        We now compare the bound (\ref{SWBoundForEA}) due to \cite{1210.4722}
        and our current bound (\ref{NSBoundForEA}) due to Theorem \ref{main-T}.
        The SDP for $B(\mc{N},K^2)$ is simply given by {\em relaxing}
        the constraint (\ref{NSABchoi}) in the SDP for $\mF{\NS}{\mc{N}}{K}$
        to an inequality.
Therefore, 
        \begin{equation}\label{relaxfact}
            \mF{\NS}{\mc{N}}{K} \leq B(\mc{N},K^2).  
        \end{equation}
        So the expression for $\mF{\NS}{\mc{N}}{K}$ given by
        Theorem \ref{main-T} gives
        an upper bound for entanglement-assisted codes at least as
        good as $B(\mc{N},K^2)$ (though we do not know if
        $\mF{\NS}{\mc{N}}{K}$ is strictly better than $B(\mc{N},K^2)$).
        Furthermore, $\mF{\NS}{\mc{N}}{K}$ is a stronger bound since it
        applies to the larger class of non-signalling codes.

        Regarding the asymptotic performance of non-signalling
        codes, it is clear that they yield a quantum capacity
        which is at least as large as the entanglement-assisted
        capacity. 

        We now argue that,
        for memoryless channels, the (asymptotic) capacities for 
        non-signalling codes and entanglement-assisted codes
        are, in fact, equal.  That is, 
        \begin{equation}
            Q^{\EA}(\mc{N}) = Q^{\NS}(\mc{N}).
        \end{equation}
        Clearly, non-signalling codes yield a quantum capacity no less
        than the entanglement-assisted capacity.
        To see the reverse, we start with a result in \cite{1210.4722} 
        showing that 
        an asymptotic analysis of $B(\mc{N}^{\ox n},K^{2n})$
        recovers the single-letter formula (\ref{BSST}),
        as an upper bound on $Q^{\EA}(\mc{N})$.
        From this result and the 
        inequality (\ref{relaxfact}), it follows that
        (\ref{BSST})
        is an upper bound even on $Q^{\NS}(\mc{N})$.
        Therefore the entanglement-assisted capacity of a memoryless
        quantum channel is \emph{equal} to the quantum capacity
        attained by non-signalling codes.
    
    \section{PPT preserving codes and distillation protocols}\label{PPTpCD}

The main result of this section, Prop.~\ref{BDSWarg}, 
relates PPT-preserving codes
and PPT-preserving entanglement distillation scheme studied in by
Rains in \cite{2001-Rains-SDP}. We will use it later
to obtain the values of  $Q^{\PPTp}$ and
$Q_{0}^{\PPTp}$ for the $d$-dimensional Werner-Holevo channel.

In \cite{2001-Rains-SDP}, Rains considers entanglement distillation by
PPT-preserving operations. He studies the quantity
\begin{equation}
    \begin{split}
    \PPTf(\rho_{\aoc\bi},K)
    :=&
    \max \{ \tr \phi_{\aic\bo}
    \distop \rho_{\aoc\bi} :\\
    &\distop \text{ is PPT-preserving},\\
    &\dim \tilde{A} = \dim \bo = K \}
    \end{split}
\end{equation}
which is the optimal entanglement fidelity
of $K \times K$ states that can be obtained from
$\rho_{\aic\bo}$ by PPT-preserving operations.
(We use these system labels to be consistent with those used later
in this section.) He also defines an associated asymptotic rate of
distillation
\begin{equation}
    \PPTD(\rho_{\aoc\bi}) :=
    \sup \{ r : \lim_{n \to \infty}
    \PPTf(\rho_{\aoc\bi}^{\ox n}, \lfloor 2^{nr} \rfloor) = 1 \}. 
\end{equation}
In the following,
we borrow ideas from
\cite{1996-BennettDiVincenzoSmolinWootters} relating error
correcting codes and entanglement distillation,
to relate PPT-preserving distillation of the Choi state of
$\mc{N}_{\bi\<\ao}$ to the channel fidelity of PPT-preserving
codes over $\mc{N}_{\bi\<\ao}$.

    \begin{proposition}\label{BDSWarg}
For any channel $\mc{N}_{\bi\<\ao}$, let 
\begin{equation}\label{N-choi}
    \nu_{\bi\aoc} := \mc{N}_{\bi\<\ao}\phi_{\ao\aoc}
    = \id_{\aoc\<\ao}N_{\bi\ao}/\dim(\ao)
\end{equation}
denote its Choi \emph{state}.

        (i) If a PPT-preserving operation
        can distills a $K \times K$ state
        from $\nu_{\bi\aoc}$
        with entanglement fidelity $f$,
        then there is a PPT-preserving code
        of size $K$ and channel fidelity $f$
        for $\mc{N}_{\bi\<\ao}$.
        Therefore, 
        $\mF{\PPTp}{\mc{N}}{K} \geq \PPTf(\nu_{\bi\aoc},K)$
        and
        $Q^{\PPTp}(\mc{N}) \geq \PPTD(\nu_{\bi\aoc})$.

        (ii) If $\mc{N}_{\bo\<\ai}$ can be implemented
        exactly using a single copy of its Choi state
        $\nu_{\bi\aoc}$
        and forward classical communication,
        then the converse to (i) is also true,
        and therefore
        $\mF{\PPTp}{\mc{N}}{K} = \PPTf(\nu_{\bi\aoc},K)$ and
        $Q^{\PPTp}(\mc{N}) = \PPTD(\nu_{\bi\aoc})$.
    \end{proposition}

If the condition for (ii) holds, Rains' SDP for the PPT fidelity for
$\mc{N}_{\bo\<\ai}$ yields a special case of Theorem \ref{main-T}.

    \begin{proof}

    (i)
    Suppose that there is a PPT preserving distillation operation
    $\distop$ which takes the Choi state
    $\nu_{\bi\aoc}$
    to a state with entanglement fidelity $f$.
    As noted by Rains,
    this fidelity is unchanged if $\distop$
    is followed by the twirling operation
    $\mc{T}_{\aic\bo\<\aic\bo}$ with a definition similar 
    to that in (\ref{twirling-op}).
    So, the operation 
    $\mc{T}_{\aic\bo\<\aic\bo} \distop$,
    has the same fidelity for input
    $\nu_{\bi\aoc}$,
    and remains PPT preserving,
    but is also non-signalling in both directions.
    This is simply because the marginal
    state of each party's system after twirling
    is always a maximally mixed state, independent of the input.
    Altogether, without loss of generality, $\distop$
    can be chosen to be non-signalling in both directions. 
    
    \begin{figure}[h]
        \centering
        \input{dist2code.tex}
        \caption{Building a PPT-preserving code (the operations in the 
          dotted box) 
          based on a PPT-preserving distillation protocol 
          (the dark grey operations and the dashed line).}
        \label{fig:dist2code}
    \end{figure}
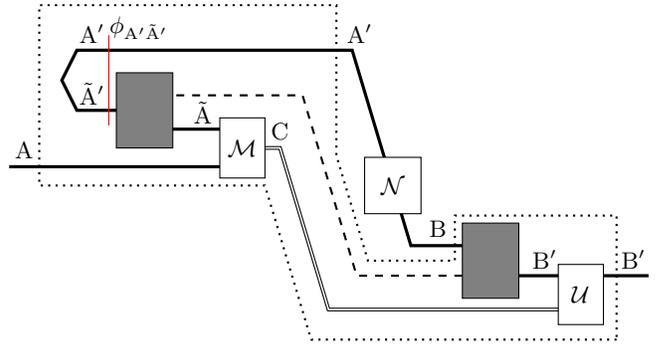

    We now construct a PPT-preserving code of dimension $K$ that is
    non-signalling from Bob to Alice using $\distop$.  
    Conceptually, the construction is the composition of three operations.
    First, Alice locally prepares the state $\phi_{\ao\aoc}$ and sends
    $\ao$ to Bob using $\mc{N}_{\bi|\ao}$ so they share the Choi state
    $\mc{N}_{\bi|\ao} \phi_{\ao\aoc}$.
    Second, they apply $\distop$ to distill a state
    $\psi_{\aic\bo}$ with channel fidelity $f$.
    Finally, Alice teleports a $K$-dimensional system from $\ai$ to
    $\bo$ using $\psi_{\aic\bo}$ instead of $\phi_{\aic\bo}$.  The 
    teleportation has channel fidelity $f$.

    These three steps are shown in Fig.\ \ref{fig:dist2code}.  
    Since $\distop$ is non-signalling from
    Bob to Alice, it can be implemented by local operations (the grey
    boxes in Fig.\ \ref{fig:dist2code}) and quantum communication from
    Alice to Bob (represented by the dashed line in
    Fig.\ \ref{fig:dist2code}). This is significant, because
    it means that Alice can complete all of her local operations
    before Bob starts his. The teleportation procedure consists
    of Alice's local measurement $\mc{M}_{\C\<\ai\aic}$, forward
    classical communication of system $\C$, and Bob's locally controlled
    unitary $\mc{U}_{\bo\<\bo\C}$.
    
    The PPT-preserving code is derived from Fig.\ \ref{fig:dist2code}
    with the encoder (decoder) being all of Alice's (Bob's) local
    operations combined, and the forward side channel being the
    communication of $\C$ combined with the forward channel in
    $\distop$.  
    Used with $\mc{N}_{\bi\<\ao}$, the code effects
    the same transmission from $\ai$ to $\bo$ as in the conceptual
    composition described earlier.
    The forward-assisted code has size $K$ and bipartite operation
    $\mc{Z}_{\ao\bo\<\ai\bi} = \mc{U}_{\bo\<\bo\C}\mc{M}_{\C\<\ai\aic}
    \distop \phi_{\ao\aoc}$.
    Since $\mc{Z}_{\ao\bo\<\ai\bi}$ is the composition of the
    PPT-preserving $\distop$ and the (one-way) LOCC operation
    $\mc{U}_{\bo\<\bo\C}\mc{M}_{\C\<\ai\aic}$, the code is 
    PPT-preserving.

    \begin{figure}[h]
        \centering
        \input{code2dist.tex}
        \caption{Building a PPT-preserving distillation operation
        from a PPT-preserving code.}\label{fig:code2dist}
    \end{figure}
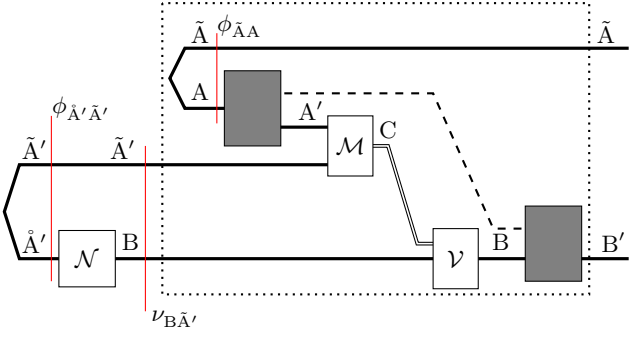

    For part (ii), suppose the channel $\mc{N}_{\bi\<\ao}$ can be
    simulated exactly using a shared copy of its Choi state and
    forward classical communication.  
    Referring to Figure \ref{fig:code2dist}, this means that
    $\mc{V}_{\bi\<\bi\C}\mc{M}_{\C\<\aoc\ao}
    \nu_{\bi\aoc}
    = \mc{N}_{\bi\<\ao}$.
    Let $\mc{Z}_{\ao\bo\<\ai\bi}$ be the bipartite operation
    corresponding to a forward-assisted code which,
    transmits a $K$-dimensional state over $\mc{N}_{\bi\<\ao}$
    with channel fidelity $f$.
    If one composes the channel simulation with $\mc{Z}_{\ao\bo\<\ai\bi}$
    as in figure \ref{fig:code2dist}, the operations in the dashed
    box distills the Choi state with fidelity $f$.
    Furthermore if $\mc{Z}_{\ao\bo\<\ai\bi}$ is PPT-preserving and 
    non-signalling from Bob to Alice, so is the distillation 
    operation.  
    \end{proof}

    It would be useful to have
    complete characterisation of channels
    that can be implemented exactly using a single copy of their Choi
    states and forward classical communication.
    We are unaware of such a characterisation in the literature.
    Here, we give a sufficient condition for a channel to 
    have this property.
    Let $\aocc$ be a copy of system $\ao$,
    and let us write
    \begin{equation}
        \nu_{\bi\aoc} =
        \mc{N}_{\bi\<\aocc} \phi_{\aocc\aoc}
    \end{equation}
    (as shown in
    Figure~\ref{fig:code2dist})
    where $\mc{N}_{\bi\<\aocc} := \mc{N}_{\bi\<\ao}
    \id_{\ao\<\aocc}$.
    Suppose that we choose the measurement operation
    $\mc{M}_{\C\<\ao\aoc}$
    and a controlled unitary operation
    $\mc{U}_{\aocc\<\aocc\C}$
    so that they comprise a teleportation protocol, such that
    \begin{equation}\label{teleports}
        \mc{U}_{\aocc\<\aocc\C}
    \mc{M}_{\C\<\ao\aoc}
    \phi_{\aocc\aoc} = \id_{\aocc\<\ao}.
    \end{equation}
    Here, $\mc{U}_{\aocc\<\aocc\C}$
    measures system $\C$ in the computational
    basis, obtaining an outcome $i$, and then applies
    a unitary transformation
    $\mc{U}^{(i)}_{\aocc\<\aocc}$
    to system $\aocc$.
    
    Now, suppose that there are unitary operations
    $\mc{V}^{(i)}_{\bi\<\bi}$ for each $i$
    such that
    \begin{equation}\label{covar}
        \forall i:~
        \mc{N}_{\bi\<\aocc}
    \mc{U}^{(i)}_{\aocc\<\aocc}
    = \mc{V}^{(i)}_{\bi\<\bi}\mc{N}_{\bi\<\aocc}.
    \end{equation}
    Let $\mc{V}_{\bi\<\bi\C}$ be a controlled unitary
    which measures $\C$ in the computational basis,
    and applies $\mc{V}^{(i)}_{\bi\<\bi}$ on obtaining
    outcome $i$.
    Then, using (\ref{covar}) and (\ref{teleports}),
    \begin{equation}
        \begin{split}
            &\mc{V}_{\bi\<\bi\C}
    \mc{M}_{\C\<\ao\aoc}
    \nu_{\bi\aoc}\\
        =&\mc{V}_{\bi\<\bi\C}
    \mc{M}_{\C\<\ao\aoc}
    \mc{N}_{\bi\<\aocc}
    \phi_{\aocc\aoc}\\
    =&
    \mc{N}_{\bi\<\aocc}
    \mc{U}_{\aocc\<\aocc\C}
    \mc{M}_{\C\<\ao\aoc}
    \phi_{\aocc\aoc}\\
    =&
    \mc{N}_{\bi\<\ao}.
        \end{split}
    \end{equation}
    That is, a use of $\mc{N}_{\bi\<\ao}$ can be implemented
    by a single copy of its Choi state $\nu_{\bi\aoc}$,
    local operations and forward classical communication.

%%%%%%%%%%%%%%%%%%%%%%%%%%%%%%%%%%%%%%%%%%%%%%%%%%%%%%%%%%%%%%%%%%%%%%%%%%%%%    
    \section{Coding over generalised Werner-Holevo channels}\label{WH}

    In this section, we apply the SDPs developed in section \ref{SDPs}
    to investigate the performance of codes which are non-signalling,
    PPT-preserving, or both, over the generalised Werner-Holevo       
    channels~\cite{2002-WernerHolevo}.
    For each dimension $d \geq 2$, consider the one-parameter family 
    of channels 
    \begin{equation}\label{GWH}
        \mc{W}^{(d,\alpha)}_{\bi\<\ao} := 
        (1-\alpha) \mc{W}^{(d,0)}_{\bi\<\ao} +
        \alpha \mc{W}^{(d,1)}_{\bi\<\ao},
    \end{equation}
    where
    \begin{align}
        \mc{W}^{(d,1)}_{\bi\<\ao}:& X_{\ao} \mapsto
        \frac{1}{d-1} (\1_{\bi} \tr X - \id_{\bi|\ao}X_{\ao}^{\T}), \text{ and}\\
        \mc{W}^{(d,0)}_{\bi\<\ao}:& X_{\ao} \mapsto
        \frac{1}{d+1} (\1_{\bi} \tr X + \id_{\bi|\ao}X_{\ao}^{\T}).
    \end{align}
    $\mc{W}^{(d,1)}_{\bi\<\ao}$ is often called
    the d-dimensional \emph{Werner-Holevo channel}.  
    Recall that $\sym_{\bi\ao}$ and $\asym_{\bi\ao}$ denote the
    projectors onto the symmetric and the antisymmetric subspaces of
    $\mc{H}_{\bi}\ox\mc{H}_{\ao}$ respectively.
    The Choi matrices of $\mc{W}^{(d,1)}_{\bi\<\ao}$ and 
    $\mc{W}^{(d,0)}_{\bi\<\ao}$ are proportional to $\asym_{\bi\ao}$ 
    and $\sym_{\bi\ao}$ respectively.

    The three-dimensional Werner-Holevo channel
    $\mc{W}^{(3,1)}_{\bi\<\ao}$
    has a Stinespring representation
\begin{align}
    \mc{W}^{(3,1)}_{\bi\<\ao}:& X_{\ao}
    \mapsto \tr_{\bs{E}} V X_{\ao} V^{\dagger},
\nonumber 
\\
    V :=& 2^{-1/2} \sum_{i,j,k=1}^{3}
\varepsilon_{ijk} \ket{j}_{\bi}\ket{k}_{\bs{E}} \bra{i}_{\ao}
\label{3dwh}
\end{align}
where $\varepsilon_{ijk}$ is the three-dimensional Levi-Civita symbol,
which is $1$ when $ijk$ is an even permutation of $123$, $-1$ when
$ijk$ is an odd permutation of $123$ and $0$ otherwise.
From (\ref{3dwh}), we see that $\mc{W}^{(3,1)}_{\bi\<\ao}$
is \emph{symmetric}, meaning that
\begin{equation}
    \tr_{\bs{E}} V X_{\ao} V^{\dagger}
    = \id_{\bs{B}\<\bs{E}} \tr_{\bs{B}} V X_{\ao} V^{\dagger}.
\end{equation}
Therefore $\mc{W}^{(3,1)}_{\bi\<\ao}$ is anti-degradable 
and hence has no unassisted quantum capacity; i.e. $Q^{\UA}(\mc{W}^{(3,1)}_{\bi\<\ao}) = 0$.

The quantum Lov\'{a}sz bound of Duan,
Severini and Winter \cite{2013-DuanSeveriniWinter}
is easily applied to this channel to
establish that it has no zero-error classical capacity, even with
arbitrary entanglement assistance.

By its definition, the generalised Werner-Holevo channels
have the covariance property that, for all unitary operations
$\mc{U}_{\ao\<\ao}: X_{\ao} \mapsto U_{\ao} X_{\ao} U^{\dagger}_{\ao}$
(where $U_{\ao}$ is a unitary operator on $\Hs_{\ao}$),
we have
\begin{equation}
    \mc{W}^{(d,\alpha)}_{\bi\<\ao} \mc{U}_{\ao\<\ao}
    =
    \mc{V}_{\bi\<\bi} \mc{W}^{(d,\alpha)}_{\bi\<\ao}
\end{equation}
where $\mc{V}_{\bi\<\bi}:
X_{\bi} \mapsto U^{\T}_{\bi} X_{\bi} \bar{U}_{\bi}$
and $U_{\bi} := \id_{\bi\<\ao} U_{\ao}$.
By the argument at the end of section \ref{PPTpCD}, 
$n$ uses of $\mc{W}^{(d,\alpha)}_{\bi\<\ao}$ can be exactly simulated using $n$
copies of the corresponding Choi state and forward classical
communication by teleportation.
    Therefore, by Proposition \ref{BDSWarg},
    the performance of PPT-preserving codes
    over these channels corresponds exactly to the
    performance of PPT-preserving distillation protocols
    on the corresponding Choi states studied 
    by Rains \cite{2001-Rains-SDP}.

    Corollary 5.6 of Rains \cite{2001-Rains-SDP}
    shows that PPT-preserving 
    operations can distill entanglement from multiple
    copies of $\mc{W}^{(d,1)}_{\bi\<\ao}$ at an optimal rate of 
    $\log ((d+2)/d)$
    ebits per state, asymptotically.  Furthermore, this rate 
    is achieved for \emph{exact} distillation. 
    Thus the quantum capacity and the
    \emph{zero-error} quantum capacity of PPT-preserving codes over
    $\mc{W}^{(d,1)}_{\bi\<\ao}$ are both $\log (d+2)/d$.
    \begin{equation}
        Q^{\PPTp}(\mc{W}^{(d,1)})
        = Q_{0}^{\PPTp}(\mc{W}^{(d,1)})
        = \log \frac{d+2}{d}.
        \label{rains-wh-zero-error}
    \end{equation}
    
    On the other hand, using the result of
    Section \ref{NScodes}
    and Eq.~(\ref{BSST}) one finds
    \begin{equation}
        Q^{\NS}(\mc{W}^{(d,1)})
    = Q^{\EA}(\mc{W}^{(d,1)})
    = \frac{1}{2}\log\frac{2d}{d-1}.
    \end{equation}
    
%-----------------------------------------------------------------------
    \subsection{\bf Performance of non-signalling, PPT-preserving codes for
    $10 \leq n \leq 120$ with fixed rates.}
    The generalised Werner-Holevo channel has high degree of symmetry.
    We exploit this symmetry to reduce the semidefinite programs
    described in Theorem \ref{main-T} and Corollary \ref{primal-SDP}
    to linear programs in $n+1$ (real) variables, for $n$ uses of the
    generalised Werner-Holevo channel.  (See Appendix \ref{LPs}.)
    The resulting linear programs can be stated using rational numbers,
    and we have evaluated their solutions \emph{exactly} using
    Mathematica's `LinearProgramming' function.
   
    In Figure \ref{code-performance-two},
    we plot 
    the log of the fidelity
    $F^{\PPTp\cap\NS}((\mc{W}^{(3,1)})^{\ox n}, \lfloor 2^{nr} \rfloor)$
    as a function of blocklength $n$ for the two rates
    $r \in \{\log (5/2 - 1/40), \log (5/2 - 1/20)\}$. 
    While the fidelity eventually goes to one
    at rate $\log (5/2 - 1/20)$,
    it appears to exhibit an exponential
    decay at rate $\log (5/2 - 1/40)$.
    
    \begin{figure}[ht]
        \centering
        \includegraphics[%bb = 0 0 243 157,
        scale=1.2]{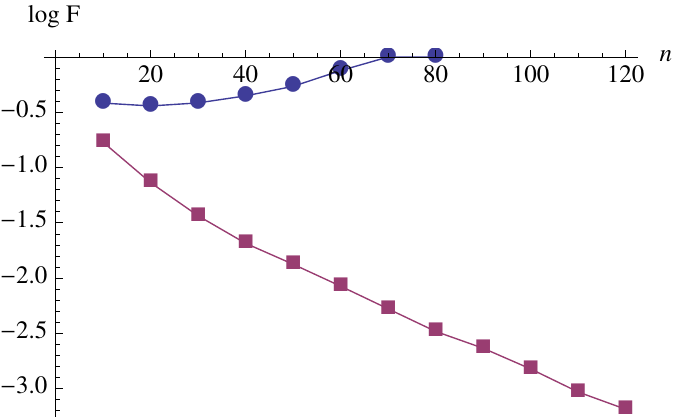}
        \caption{The logarithm (base-two)
        of the optimal channel fidelity for 
        non-signalling, PPT-preserving codes of size
        $K_n = \lfloor 2^{r n} \rfloor$
        for $n$ uses of the three-dimensional Werner-Holevo
        channel at rates
        $r = \log (5/2 - 1/20)$ (circles) and
        $r = \log (5/2 - 1/40)$ (squares).}
        \label{code-performance-two}
   \end{figure}

    From Eq.~(\ref{rains-wh-zero-error}), at either rate studied
    above, the fidelity for PPT-preserving codes is $1$.  Thus our
    non-signalling and PPT-preserving codes provides strictly tighter
    bound for the unassisted performance for finite block-length.

    If the code fidelity
    $F^{\PPTp\cap\NS}((\mc{W}^{(3,1)})^{\ox n}, \lfloor 2^{nr} \rfloor)$
    at rate $\log (5/2 - 1/40)$ does not eventually 
    increase and approach $1$ as $n$ increase,
    then $Q^{\PPTp\cap\NS}(\mc{W}^{(3,1)})$
    is no more than $\log (5/2 - 1/40)$,
    which is strictly less than both
    $Q^{\PPTp}(\mc{W}^{(3,1)}) = \log (5/2)$
    and $Q^{\NS}(\mc{W}^{(3,1)}) = \frac{1}{2}\log 3$.
    
    Deciding whether there really can be
    a separation between the asymptotic capacities
    $Q^{\PPTp}$ and $Q^{\PPTp\cap\NS}$
    presents an interesting open problem.

%-----------------------------------------------------------------------
    \subsection{\bf Performance of non-signalling, PPT-preserving codes for
    $n=2$ with variable rate.}

    In Figure \ref{code-performance-one} we plot the channel code
    fidelities when the channel operation is two uses of the three
    dimensional Werner-Holevo channel.  We consider codes that 
    are non-signalling, PPT-preserving, or both.

\begin{figure}[ht]
\centering
        \includegraphics[bb = 0 1 198 131,scale=1.0]{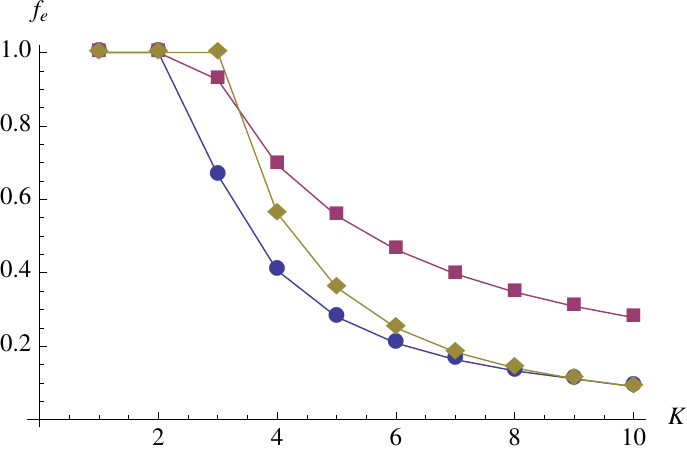}
        \caption{The optimal channel fidelity for sending
        the state of a $K$-dimensional system over two-uses
        of the three-dimensional Werner-Holevo channel using
        a code which is (i) non-signalling (yellow diamonds),
        (ii) PPT-preserving (red squares)
        (iii) both non-signalling and PPT-preserving (blue circles).}
        \label{code-performance-one}
    \end{figure}

    First, we note that $\mF{\PPTp}{\mc{N}}{K}$ and
    $\mF{\NS}{\mc{N}}{K}$ are incomparable.  In particular,
    non-signalling codes can transmit a $3$-dimensional system 
    with fidelity one, but has fidelity lower than that of PPT-preserving 
    codes for $K \geq 4$.  Also, $\mF{\PPTp \cap \NS}{\mc{N}}{K}$
    is strictly less than
    $\min \{\mF{\PPTp}{\mc{N}}{K}, \mF{\NS}{\mc{N}}{K}\}$,
    for small $n$ and therefore provides a better bound on 
    $\mF{\UA}{\mc{N}}{K}$.
    Finally, for $K\geq 9$, non-signalling codes can be chosen to 
    be also PPT preserving without affecting the fidelity.

%-----------------------------------------------------------------------
    \subsection{\bf A bonus observation -- superactivation?}

    Consider the specific data point $K=2$ in Figure \ref{code-performance-one}. 
    All three curves coincide at this point and have value $1$.  
    Thus zero-error quantum communication of a qubit 
    is possible over the channel even if we demand
    that the code be both \emph{non-signalling} and PPT-preserving.  
    While Rain's work \cite{2001-Rains-SDP} already implies the 
    possibility given PPT-preserving codes, 
    it is somewhat surprising that one can further restrict to 
    non-signalling codes.  
    \begin{example}\label{example}
        There is a PPT-preserving, non-signalling code
        which can transmit a qubit with perfect entanglement
        fidelity over two uses of the three dimensional Werner-Holevo
        channel: The channel input system is
        $\ao = \ao_{1}\ao_{2}$
        and the channel output system is
        $\bi = \bi_{1}\bi_{2}$, and the channel operation is
        $\mc{W}^{(3,1)}_{\bi_{1}\<\ao_{1}}
        \mc{W}^{(3,1)}_{\bi_{2}\<\ao_{2}}$.
        The code is given by taking the
        maximally mixed average channel input
        $\rho_{\ao} = \1_{\ao}/\dim(\ao)$
        and choosing
        $\sdpop_{\ao\bi}
        = \frac{3}{32} \sym_{\ao_{1}\bi_{1}} \sym_{\ao_{2}\bi_{2}}
        +\frac{7}{32} (\sym_{\ao_{1}\bi_{1}} \asym_{\ao_{2}\bi_{2}}
        + \asym_{\ao_{1}\bi_{1}} \sym_{\ao_{2}\bi_{2}})
         + \asym_{\ao_{1}\bi_{1}} \asym_{\ao_{2}\bi_{2}}$
        in the expressions (\ref{NSBAtwirled})
        and (\ref{U-invar-form}).
    \end{example}
        
As discussed in Section \ref{code-classes},
PPT-preserving codes include codes assisted by
arbitrary forward communication over Horodecki channels.
Therefore, the results of Smith and Yard on superactivation
\cite{2008-SmithYard} mean that such codes may yield quantum capacity
over symmetric channels. Nevertheless, we were somewhat surprised
to find that a PPT-preserving non-signalling code allows the perfect
transmission of a single qubit over two uses of a simple
example of a symmetric channel.

Since the code operation is non-signalling from Bob to Alice,
it can be implemented by forward quantum communication from
Alice to Bob. This is the result of Eggeling and Schlingemann
and Werner \cite{2002-EggelingSchlingemannWerner},
that ``semicausal operations are semilocalisable''.
The use of this forward quantum communication is somehow ``hidden''
by the local operations performed by Alice and Bob in the implementation
so that the resulting bipartite operation is both PPT and non-signalling.

Given the result of Eggeling et al.,
it might be tempting to guess that a bipartite operation which
is non-signalling from Bob to Alice \emph{and} PPT-preserving, like the forward-assisted code in Example $\ref{example}$
(which is also non-signalling from Alice to Bob),
can always be implemented by forward communication \emph{over a
Horodecki channel}.
If this were possible for our Example \ref{example},
or for some other PPT-preserving code enabling
zero-error quantum communication over a channel without quantum capacity
then it would constitute a remarkably extreme version of the superactivation
phenomenon discovered by Smith and Yard \cite{2008-SmithYard}.
We leave this question open here. However we can give an example
which shows that this kind of implementation is not always possible,
even when the bipartite operation is non-signalling in both directions.

The example is a bipartite operation $\mc{Z}_{\ao\bo\<\ai\bi}$ with
$\dim \ai = \dim \ao = \dim \bi = \dim \bo = 2$, which we
will describe by giving a particular protocol to implement the operation,
which is illustrated in the top half of Figure \ref{NSPPTneqFHA}.
Bob measures his input $\bi$ in the computational basis and sends
the outcome $b$ to Alice. He also generates an unbiased random bit $r$,
which he sends to Alice and outputs on $\bo$ in the computational basis.
If $b=0$ Alice does nothing, but if $b=1$ she applies a Hadamard gate
to $\ai$. Then, regardless of the value of $b$,
she measures $\ai$ in the computational basis yielding outcome $a$.
She outputs $a \oplus r$ on $\ao$ in the computational basis.

Since the operation can be implemented using only
classical communication from Bob to Alice, it is certainly
a PPT-preserving measurement.
The marginal states of $\ao$ and $\bo$ are both maximally mixed states,
independent of the input state, so the operation is non-signalling in
both directions.

\newcommand*{\gcrad}{0.15}
\def\aline{1}
\def\bline{-1}
\def\leftx{-0.3}
\def\rightx{6}

\def\bmx{1}
\def\ahx{2.5}
\def\amx{3.5}
\def\brx{3.5}
\def\cnotx{4.4}

\def\gateR{0.4}
\def\lo{0.2}
\def\og{0.2}
\newcommand{\gate}[3]{\draw [fill=white] (#1-\gateR,#2-\gateR)
            rectangle (#1+\gateR,#2+\gateR); \node at (#1,#2) {#3};}
\newcommand{\xor}[2]{\draw (#1,#2) circle (\gcrad);
            \draw[style=very thick] (#1-\gcrad,#2) -- (#1+\gcrad,#2);
            \draw[style=very thick] (#1,#2-\gcrad) -- (#1,#2+\gcrad);}
\newcommand{\ctl}[2]{\draw[fill=black] (#1,#2) circle (\gcrad);}
\begin{figure}
    \centering
    \begin{tikzpicture}[scale=1.0] %% LOCC operations
        %% Alice's line
        \draw[draw=black,very thick]
        (\leftx,\aline) -- (\rightx,\aline);

        %% Bobs's line
        \draw[draw=black,very thick]
        (\leftx,\bline) -- (\bmx,\bline);

        %% classical communication
        \draw[draw=black,style=double]
        (\bmx,\bline+\lo) --
        (\bmx+\gateR+\lo,\bline+\lo) --
        (\ahx-\gateR-\lo,\aline-\lo) -- (\ahx,\aline-\lo);

        % bob's output line
        \draw[draw=black,style=very thick]
        (\brx,\bline) -- (\rightx,\bline);

        %% Controled NOT
        \draw[draw=black,style=double]
        (\cnotx,\aline) -- (\cnotx,\bline);
        \xor{\cnotx}{\aline};
        \ctl{\cnotx}{\bline};

        %% operations
        \gate{\bmx}{\bline}{$\mc{M}$}
        \gate{\ahx}{\aline}{$\mc{H}$}
        \gate{\amx}{\aline}{$\mc{M}$}
        \gate{\brx}{\bline}{$\mc{P}$}

        \draw[draw=black,dashed]
        (\bmx-\gateR-\og,\aline+\gateR+\og)
        rectangle (\cnotx+\gateR+\og,\bline-\gateR-\og);

        \node[text=black] (Z) at (0.7,0) {$\Zc$};

        % System labels
        \node at (\leftx+0.3,\aline+\slo) {$\ai$};
        \node at (\leftx+0.3,\bline+\slo) {$\bi$};
        \node at (\rightx-0.3,\aline+\slo) {$\ao$};
        \node at (\rightx-0.3,\bline+\slo) {$\bo$};
    \end{tikzpicture}
   \hspace{8ex}
   \def\ahx{4}
   \def\encx{1.6}
   \def\rightx{7}
   \begin{tikzpicture}[scale=1.0] %% LOCC operations
        \def\cnotwo{6}
        %% Alice's line
        \draw[draw=black,very thick]
        (\leftx,\aline) -- (\cnotx+1.0,\aline) --
        (\cnotwo,\bline+0.4) -- (\cnotwo,\bline);
        
        \gate{0.5*\cnotx+0.5+0.5*\cnotwo}{0}{$\mc{C}$}

        %% Bobs's line
        \draw[draw=black,very thick]
        (\leftx,\bline) -- (\rightx,\bline);

        %% forward communication
        \draw[draw=black, very thick]
        (\encx,\aline-\lo) --
        (\encx+\gateR+\lo,\aline-\lo) --
        (\ahx-\gateR-\lo,\bline+\lo) --
        (\ahx,\bline+\lo);

        % Controled NOT
        \xor{\cnotwo}{\bline};

        %% operations
        \gate{\encx}{\aline}{$\mc{E}$}
        \gate{\ahx}{\bline}{$\mc{D}$}
        \def\auxx{0.5*\encx+0.5*\ahx}
        \def\auxy{0.5*\aline+0.5*\bline}
        \def\auxw{0.5}
        \gate{\auxx}{\auxy}{$\mc{F}$}

        \draw[draw=black,dashed]
        (\bmx-\gateR-\og,\aline+\gateR+\og)
        rectangle (\cnotx+\gateR+\og,\bline-\gateR-\og);
        
        \draw[draw=black,dotted]
        (\encx-\gateR-0.1,\aline+\gateR+0.1)
        rectangle (0.5*\cnotx+0.5+0.5*\cnotwo+\gateR+0.1,
        -\gateR-0.1);
        
        \node[text=black] (G) at
        (0.5*\cnotx+0.5+0.5*\cnotwo+\gateR-0.3,\aline) {$\mc{G}$};

        \node[text=black] (Z) at (0.7,0) {$\Zc$};

         % System labels
        \node at (\leftx+0.3,\aline+\slo) {$\ai$};
        \node at (\leftx+0.3,\bline+\slo) {$\bi$};
        \node at (\cnotx+0.9,\aline+\slo) {$\ao$};
        \node at (\cnotx+0.9,\bline+\slo) {$\bo$};
        
        \node at (\auxx-0.3,\auxy+\auxw+\slo) {$\Q$};
        \node at (\auxx+0.3,\auxy-\auxw-\slo) {$\R$};
        
        \node at (\cnotx+1.8,\bline+\slo+\slo) {$\C$};
    \end{tikzpicture}
    \caption{}\label{NSPPTneqFHA}
\end{figure}

However, in the implementation just described
the communication was in the ``backward'' direction - from Bob to Alice.
We claim that implementing the operation by \emph{forward}
communication only, requires at least one qubit of
\emph{zero-error} quantum communication, which clearly
cannot be accomplished by any Horodecki channel.
Here is a proof:
The most general implementation with only forward communication
has the form
\[
    \mc{Z}_{\ao\bo\<\ai\bi} = \mc{D}_{\bo\<\bi\R}
    \mc{F}_{\R\<\Q}\mc{E}_{\ao\Q\<\ai}
\]
where $\mc{E}_{\ao\Q\<\ai}$ is Alice's local operation,
$\mc{F}_{\R\<\Q}$ is the channel used for forward communication, $\mc{D}_{\bo\<\bi\R}$ is Bob's local operation. We illustrate
this in the bottom half of Figure \ref{NSPPTneqFHA}.
Now, if Alice sends her bit $a \oplus r$ to Bob with one use of
a forward completely dephasing channel $\mc{C}_{\C\<\ao}$, then Bob can XOR $a \oplus r$ with $r$ to
obtain the outcome of Alice's measurement of the $\ai$ system.
Therefore, by measurement of the output of the operation
$\mathcal{G}_{\C\R\<\ai}
:= \mc{C}_{\C\<\ao} \mc{F}_{\R\<\Q}\mc{E}_{\ao\Q\<\ai}$
(outlined by the dotted line in Figure \ref{NSPPTneqFHA}),
Bob can choose to discriminate perfectly between
$\ket{0}_{\ai}$ and $\ket{1}_{\ai}$ or
between $\ket{+}_{\ai}$ and $\ket{-}_{\ai}$ depending on his input.
It must therefore be that
\begin{align*}
    \tr_{\C\R} (\mathcal{G}_{\C\R\<\ai} \ketbra{0}{0}_{\ai})(\mathcal{G}_{\C\R\<\ai}
    \ketbra{1}{1}_{\ai}) =& 0,\\
    \tr_{\C\R} (\mathcal{G}_{\C\R\<\ai} \ketbra{+}{+}_{\ai})(\mathcal{G}_{\C\R\<\ai}
    \ketbra{-}{-}_{\ai}) =& 0.
\end{align*}
By Lemma 1 of Cubitt and Smith \cite{2012-CubittSmith},
this implies that $\mathcal{G}$ is capable of
sending a single qubit perfectly. Since the forward classical
communication over $\mc{C}$ cannot increase the zero-error quantum
capacity of $\mc{F}$, it must be that $\mc{F}$ itself can send a single
qubit perfectly. Clearly, no Horodecki channel can do this.

\section{Conclusion}
We have shown how a number of operationally relevant
classes of quantum code (such as unassisted codes, entanglement-assisted codes,
codes assisted by forward classical communication)
can be regarded as sub-classes
of the \emph{forward-assisted codes}, which
correspond to deterministic supermaps or,
equivalently, to bipartite operations
which are non-signalling from Bob to Alice.
By requiring additionally that
these operations are PPT-preserving, non-signalling
(from Alice to Bob), or both,
we obtain non-trivial bounds on the performance
of the operationally defined classes of codes,
in the form of simple semidefinite programs.

The SDP for non-signalling codes gives an upper
bound on entanglement-assisted codes which is at least
as tight as the one given in \cite{1210.4722},
and we use this fact to show that the capacity
of entanglement-assisted and non-signalling
codes is the same for memoryless channels. 
It would be interesting to find out 
if the SDP for non-signalling codes 
is strictly better than the bound in \cite{1210.4722}.  

In the case of codes which are PPT-preserving,
we described how these are related to the
PPT-preserving entanglement distillation
protocols studied by Rains. 
This gave us a general lower bound on
the PPT-preserving code performance
and an equality between code performance
and distillation fidelity of the Choi state
for some special channels. 
This equality let us use Rains' results 
to obtain the PPT-preserving code capacities
(even the zero-error capacities) for
the $d$-dimensional Werner-Holevo channels.
In regarding the conditions for equality,
we would be interested to know a complete
characterisation of when a channel can be
implemented exactly using a single copy
of its Choi state and forward classical
communication.

By imposing both non-signalling and PPT-preserving constraints
we obtain bounds on the fidelity of unassisted quantum codes.
Again using the example of Werner-Holevo channels, we show that
this provides a strictly stronger bound in the finite block length regime.
Numerics suggest that it can may even be stronger asymptotically.  
It would be interesting to find out whether this is indeed the case; 
for example, to show a separation of the capacities, it suffices to 
find feasible solutions for the dual programs for an infinite
sequence of block lengths, which yield an upper bound
for which an asymptotic separation can be proven.

Even with both constraints on the code we find that
zero-error communication of a qubit is possible
given two uses of the three-dimensional Werner-Holevo channel.
It is not clear to us whether assistance by
forward communication over Horodecki channels would allow the
same phenomenon via ``superactivation'' in the sense
of Smith and Yard. We have given an example showing
that not \emph{all} non-signalling,
PPT-preserving bipartite operations can by implemented
by forward communication over Horodecki channels,
but this does not settle the question.
It would be of interest to do so.

One potential application of our SDP concerns bounds on regular
quantum error correcting codes.  Consider a family of channels
parameterized by some error strength $\epsilon$.  For each block
length $n$ and codespace dimension $K$, our SDP can be used to
evaluate the code fidelity as a function of $\epsilon$.  The existence
of an unassisted code that corrects for $t$ errors will imply an 
assisted code fidelity that is at least $1-O(\epsilon^{t{+}1})$.  
Thus, a numerically obtained fidelity worse than 
$1-O(\epsilon^{t{+}1})$ can be viewed as evidence for the 
non-existence of such unassisted codes.

Finally (as mentioned above)
we have determined that $Q^{\NS}$ is equal to
$Q^{\EA}$ for which there is a single-letter formula
due to \cite{BSST}, but do any of the capacities
$Q^{\FHA}$, $Q^{\NS \cap \PPTp}$ or $Q^{\PPTp}$
have a single-letter formula?

\section*{Acknowledgements}

We thank Nilanjana Datta, Runyao Duan,
Michael Wolf and Andreas Winter for insightful discussions.

%\appendices %FOR IEEEtran class
\appendix
\section{Linear program for generalised Werner-Holevo channels}\label{LPs}
We consider $n$ uses of the generalised Werner-Holevo channel
(\ref{GWH}). The input system is $\mathbf{A}' = \ao_1 \cdots \ao_n$
and the output system is $\mathbf{B} = \bi_1 \cdots \bi_n$,
where $\dim(\ao_i) = \dim(\bi_i) = d$.
The Choi matrix of the operation is $d^n w(d,\alpha)^{\ox n}$ where
$w(d,\alpha)_{\ao\bi}$ is the Werner state
$w(d,\alpha)_{\ao\bi} =
(1-\alpha) \sym_{\ao\bi}/(\tr \sym) + \alpha \asym_{\ao\bi}/(\tr \asym)$.

As such, the Choi matrix is invariant under conjugation by
$U_{\bs{A}_{j}} U_{\bs{B}_{j}}$, for all unitaries $U$
and $j \in \{1,\ldots,n\}$, and invariant under permutations.
Therefore, in the semidefinite program, there is no loss of generality
in assuming that the operator $\sdpop_{\mathbf{A}'\mathbf{B}}$
possesses the same invariance,
and that $\rho_{\mathbf{A}'}$
is invariant under the restriction of these actions
to the input subsystems. Since this 
means that $\rho_{\mathbf{A}'}$ is invariant
under an arbitrary unitary transformation of any one of the $n$
input subsystems, $\rho_{\mathbf{A}'}$
can only be the maximally mixed state
$\rho_{\mathbf{A}'} = \1_{\mathbf{A}'}/d^n$.
As for $\sdpop_{\mathbf{A}'\mathbf{B}}$,
it must be a linear combination of $n+1$ orthogonal projectors
\begin{equation}\label{n-Werner-R}
    \sdpop_{\mathbf{A}'\mathbf{B}} = \sum_{k=0}^n x_{k} E^{n}_{k}
\end{equation}
where $E^{n}_{k}$ is the sum of all $n$-fold tensor products
of the operators $\sym$ and $\asym$ which contain exactly $k$
copies of $\asym$ (see Example $\ref{example}$ for an example
of an $\sdpop_{\mathbf{A}'\mathbf{B}}$ of this form for $n=2$).
The partial transpose of such an
operator is itself given by a sum of orthogonal projectors.
Let $\Upsilon^{n}_{i}$ denote the sum of all $n$-fold tensor products
of the projectors $\1-\phi$ and $\phi$ which contain exactly $k$
copies of $\phi$ e.g.
$\Upsilon^{2}_{1} =
(\1-\phi)_{\bs{A}'_{1}\bs{B}_{1}}\phi_{\bs{A}'_{2}\bs{B}_{2}} +
\phi_{\bs{A}'_1\bs{B}_1}(\1-\phi)_{\bs{A}'_2\bs{B}_2}$.
Then
\begin{equation}
    \t_{\mathbf{B}\<\mathbf{B}}
    \sdpop_{\mathbf{A}'\mathbf{B}}
    = \sum_{i,j=0}^n \Upsilon^{n}_{i} M^{(n)}_{ij} x_j,
\end{equation}
where
\begin{equation}
    M^{(n)}_{ij}
    := 2^{-n} \sum_{k=0}^{\min\{i,j\}} \binom{n-i}{j-k} \binom{i}{k}(1+d)^{i-k} (1-d)^k.
\end{equation}
See \cite{2009-MatthewsWinter} for the
derivation of this formula for $M^{(n)}_{ij}$.
For the non-signalling constraint, the fact that
$\tr_{\bi} \sym_{\ao\bi} = \frac{d+1}{2} \1_{\ao}$ and
$\tr_{\bi} \asym_{\ao\bi} = \frac{d-1}{2} \1_{\ao}$
and a little counting show that
$\tr_{\mathbf{B}} E^{n}_{j} = g^{(n)}_j \1_{\bi}$ where
\[
    g^{(n)}_j := 2^{-n} \binom{n}{j} (d+1)^{n-j} (d-1)^j.
\]
Substituting (\ref{n-Werner-R}) and $\rho_{\ao} = d^{-n}\1_{\mathbf{A}'}$
into the SDP described in Theorem \ref{main-T} and
Corollary \ref{primal-SDP} and using the facts just established,
we obtain
\begin{proposition}
The optimal channel fidelity
of a forward-assisted-code of size $K$
for $n$ uses of the $d$-dimensional
generalised Werner-Holevo channel $\mc{W}(d,\alpha)$ is
given by the linear program
\begin{align}
    \max d^n &\sum_{j = 0}\binom{n}{j} (1-\alpha)^{n-j}\alpha^j x_j\\
    &\text{subject to }\notag\\
    &\text{ for all } i = 0,\ldots,n\\
    &0 \leq x_i \leq d^{-n}
\end{align}
with the additional constraint 
\begin{equation}
    \NS:
    \sum_{j=0}^n g^{(n)}_j x_j = 1/K^2
\end{equation}
if the code is non-signalling, and the constraint
\begin{equation}
    \PPTp:
    \begin{cases}
        \sum_{j=0}^n M^{(n)}_{ij} x_j \geq -d^{-n}/K,\\
        \sum_{j=0}^n M^{(n)}_{ij} x_j \leq d^{-n}/K,
    \end{cases}
\end{equation}
if the code is PPT-preserving.
\end{proposition}

\section{Teleportation and dense coding}\label{apdx}

\def\depol{\mc{X}}
It will be useful to define a \emph{non-signalling channel}
to be one of the form $\mc{R}_{\bi\<\ao} = \sigma_{\bi} \tr_{\ao}$.
If we apply a non-signalling code (or, more generally, a
non-signalling deterministic supermap) to a non-signalling
channel, then it is not hard to see that the result is
also a non-signalling channel.

The $d$-\emph{ary symmetric classical channels}
$\mc{C}^{(p)}_{\Q'\<\Q}$
in $\ops(\Q \to \Q')$ with $\dim \Q = \dim \Q' = d$
can be parameterised by their success probability
$p$ such that: $\mc{C}^{(1)}_{\Q'\<\Q}$
is the classical identity channel $\mc{C}_{\Q'\<\Q}$
;
$\mc{C}^{(1/d)}_{\Q'\<\Q} := d^{-1} \1_{\Q'} \tr_{\Q}$,
(which is a non-signalling channel);
and $\mc{C}^{(p)}$ defined so that $p \mapsto \mc{C}^{(p)}$
is linear in $p$. This results in $\mc{C}^{(p)}$ being
a valid operation for the range $p \in [0,1]$.
The `symmetry' in their name refers to the fact that
the channels commute with an permutation of the
computational basis elements.
The deterministic supermap $\mathfrak{P}$ defined by
\begin{equation}
    \mathfrak{P}: \mc{M}_{\Q'\<\Q}
    \mapsto \frac{1}{d!} \sum_{\pi}
    \mc{U}^{\pi^{-1}}_{\Q'\<\Q'} \mc{M}_{\Q'\<\Q} \mc{U}^{\pi}_{\Q\<\Q}
\end{equation} turns any channel
$\mc{M}_{\Q'\<\Q}$ in $\ops(\Q \to \Q')$
with $\Ps(\mc{M}_{\Q'\<\Q}) = p$
into $\mc{C}^{(p)}_{\Q'\<\Q}$.
Here, $\pi$ ranges over all permutations of the numbers $\{ 0, \ldots, d-1 \}$
and $\mc{U}^{\pi}$ is the unitary operation which permutes the computational
basis vectors according to $\pi$.

Likewise, the $d$-dimensional \emph{depolarising channels}
$\depol^{(f)}$ can be parameterised by their channel
fidelity $f$, with $\depol^{(1)}_{\Q'\<\Q} = \id_{\Q'\<\Q}$,
$\depol^{(1/d^2)}_{\Q'\<\Q} = d^{-1} \1_{\Q'} \tr_{\Q}$
and the rest so that $f \mapsto \depol^{(f)}$ is linear.
Again, this means that $\depol^{(f)}$
is a valid operation for all $f \in [0,1]$.
Given any channel $\mc{M}_{\Q'\<\Q}$ in $\ops(\Q \to \Q')$ with
$\F(\mc{M}_{\Q'\<\Q}) = f$, applying the `twirling' deterministic
supermap
\begin{equation}
    \mathfrak{U}: \mc{M}_{\Q'\<\Q}
    \mapsto \int d\mu(U)
    \mc{U}_{\Q'\<\Q'} \mc{M}_{\Q'\<\Q} \mc{U}_{\Q\<\Q}
\end{equation}
where $\mu$ is the Haar probability measure on $\mathrm{U}(d)$
and $\mc{U}$ the unitary operation which conjugates by $\mc{U}$,
will turn it into $\depol^{(f)}_{\Q'\<\Q}$.

A teleportation protocol is an
entanglement-assisted code in $\EA$ taking
$\ops(\ao \to \bi)$ to $\ops(\bi \to \ao)$
where $\dim \ao = \dim \bi = d^2$
and $\dim \ai = \dim \bo = d$.
We call the deterministic supermap $\mathfrak{T}$.
It maps the $d^2$-dimensional classical identity
channel to the $d$-dimensional quantum identity
channel.
\begin{equation}\label{teleport-id}
    \mathfrak{T}[ \C^{(1)}_{\bi\<\ao} ]
    = \id_{\bo\<\ai} = \depol^{(1)}_{\bo\<\ai}.
\end{equation}
By twirling, we can assume that the channel produced
by the teleportation protocol is a depolarising channel.

The only non-signalling
$d$-dimensional depolarising channel is $\depol^{1/d^2}$ so
\begin{equation}\label{teleport-NS}
    \mathfrak{U}\circ\mathfrak{T}[ \C^{(1/d^2)}_{\bi\<\ao} ] = \depol^{(1/d^2)}_{\bo\<\ai}.
\end{equation}
By eqn.~(\ref{teleport-id}), eqn.~(\ref{teleport-NS})
and linearity we have
\begin{equation}
    \mathfrak{U}\circ\mathfrak{T}[ \C^{(\lambda)}_{\bi\<\ao} ]
    = \depol^{(\lambda)}_{\bo\<\ai}.
\end{equation}
Therefore, given any operation $\mc{M}_{\bi\<\ao}$ with
success probability $\Ps(\mc{M}_{\bi\<\ao}) = \lambda$
we can apply the entanglement-assisted deterministic supermap
$\mathfrak{U}\circ\mathfrak{T}\circ\mathfrak{P} \in \EA$
to obtain
a depolarising channel with channel fidelity $\lambda$:
\begin{equation}
    \mathfrak{U}\circ\mathfrak{T}\circ\mathfrak{P}[\mc{M}_{\bi\<\ao}]
    = \depol^{(\lambda)}_{\bo\<\ai}.
\end{equation}

A \emph{dense-coding protocol} is an
entanglement-assisted code
$\mathfrak{D}: \ops(\ao \to \bi) \to \ops(\bi \to \ao)$
where $\dim \ao = \dim \bi = d$
and $\dim \ai = \dim \bo = d^2$,
such that $\mathfrak{D}[\depol^{(1)}_{\bi\<\ao}] = \C^{(1)}_{\bo\<\ai}$.
Using a similar argument to the above find that
\begin{equation}
    \mathfrak{P}\circ\mathfrak{D}[\depol^{(\lambda)}_{\bi\<\ao}] = \C^{(\lambda)}_{\bo\<\ai}
\end{equation}
and that, from any operation $\mc{N}_{\bi\<\ao}$ with channel fidelity
$F(\mc{N}_{\bi\<\ao}) = \lambda$ we can obtain a $d^2$-ary symmetric classical
channel with success probability $\lambda$:
\begin{equation}
    \mathfrak{P}\circ\mathfrak{D}\circ\mathfrak{U}[\mc{N}_{\bi\<\ao}]
    = \mathcal{C}^{(\lambda)}_{\bo\<\ai}.
\end{equation}
\newpage

\bibliography{bigbibcopy}

\end{document}

%% file: fig1.tex
\begin{tikzpicture}[scale=1]     

    \def\rightx{3.5}

    \def\encx{-2.375}
    \def\ency{1}
    \def\encw{0.375}
    \def\ench{0.5}

    \def\decx{2.125}
    \def\decy{-0.75}
    \def\decw{0.375}
    \def\dech{0.5}

    \def\chanx{0.375}
    \def\chany{0.375}
    \def\chanw{0.375}

    \def\auxx{-1.375}
    \def\auxy{-0.125}
    \def\auxw{0.375}

    \def\lo{0.25} %% line offset
    \def\og{0.1} %% outline gap

    %% Alice's input line
    \draw[draw=black,very thick]
    (\rightx,2.0) --
    (-3.5,2.0) --
    (-3.8,1.5) --
    (-3.5,1.0) -- (-2.75,1.0);

    % input state line and label
    \draw[draw=red] (-0.2,0.8) -- (0.6,1.2);
    \node at (1.0,1.2) {$\rho_{\ao}$};

    % purified input state line and label
    \draw[draw=red] (-3.1,0.8) -- (-3.1,2.4);
    \node at (-2.8,2.5) {$\phi_{\ai\tilde{\ai}}$};

    % output state line and label
    \draw[draw=red] (3.3,-1.5) -- (3.3,2.4);
    \node at (3.5,2.5) {$\tau_{\bo\tilde{\ai}}$};

    \node at (-3.25,1.25) {$\bs{A}$};
    \node at (-0.5,1.5) {$\bs{A}'$};

    \node at (1.25,-0.25) {$\bs{B}$};
    \node at (3,-0.5) {$\bs{B}'$};

    \node at (-1.5,0.75) {$\bs{Q}$};
    \node at (-1.25,-1) {$\bs{R}$};

    %% Bobs's output line
    \draw[draw=black,very thick]
    (2.5,-0.75) -- (\rightx,-0.75);

    \draw[draw=black,very thick]
    (\encx,\ency-\lo) -- (-1.75,\ency-\lo) -- (-1.5,0.25)
    -- (-1.25,-0.5) -- (-1.0,\decy-\lo) -- (\decx,\decy-\lo);

    \draw[draw=black,very thick]
    (\encx,\ency+\lo) -- (0,1.25) -- (0.25,0.75)
    -- (0.5,0.0) -- (0.75,\decy+\lo) -- (\decx,\decy+\lo);

    % Outline of the bipartite operation Z
    \draw[draw=blue,dashed] %% going clockwise
    (\encx-\encw-\og,\decy-\dech-\og) -- %% bottom left
    (\encx-\encw-\og,\ency+\ench+\og) -- %% top left
    (\auxx+\auxw+\og,\ency+\ench+\og) -- %
    (\auxx+\auxw+\og,\decy-\lo+\og) --
    (\decx-\decw-\og,\decy-\lo+\og) --
    (\decx-\decw-\og,\decy+\dech+\og) --
    (\decx+\decw+\og,\decy+\dech+\og) -- %% right top
    (\decx+\decw+\og,\decy-\dech-\og) -- %% bottom right
    (\encx-\encw-\og,\decy-\dech-\og);

    \node[text=blue] (Z) at (\encx-\encw+0.1,\decy-\dech+0.1) {$\mc{Z}$};

    % outline of N'
    \draw[draw=black,dotted,thick]
    (\encx-\encw-2.0*\og,\ency+\ench+2.0*\og) rectangle
    (\decx+\decw+2.0*\og,\decy-\dech-0.6);

    \node[text=black] at (0,\decy-\dech-0.3) {$\mc{M}$};

    % ENCODER
    \draw [fill=white,draw=black]
    (\encx-\encw,\ency-\ench) rectangle
    (\encx+\encw,\ency+\ench);
    \node (enc) at (\encx,\ency) {$\mc{E}$};

    % AUX. CHANNEL
    \draw [fill=white]
    (\auxx-\auxw,\auxy-\auxw) rectangle
    (\auxx+\auxw,\auxy+\auxw);
    \node (sc) at (\auxx,\auxy) {$\mc{F}$};

    % CHANNEL
    \draw [fill=white]
    (\chanx-\chanw,\chany-\chanw) rectangle
    (\chanx+\chanw,\chany+\chanw);
    \node (chan) at (\chanx,\chany) {$\mc{N}$};

    % DECODER
    \draw [fill=white]
    (\decx-\decw,\decy-\dech) rectangle
    (\decx+\decw,\decy+\dech);
    \node (dec) at (\decx,\decy) {$\mc{D}$};

\end{tikzpicture}

%% file: dist2code.tex
\begin{tikzpicture}[scale=1]
    \def\Lone{1.8}
    \def\Lfour{0.7}
    \def\Lthree{-0.75}
    \def\slope{0.3}
    \def\xa{0.1}

    \def\iny{0.25}
    \def\inxs{-4.4}

    \def\phiback{-3.5}

    \def\dleft{\phiback-0.5}

    \def\encx{-2.6}
    \def\ency{1}
    \def\encw{0.375}

    \def\measy{0.5}
    \def\measx{-1.3}
    \def\meash{0.4}
    \def\measw{0.3}

    \def\chanx{0.7}
    \def\chany{0.0}
    \def\chanr{0.375}

    \def\Ux{3.2}
    \def\Uy{-1.45}

    \def\decx{2.0}
    \def\decy{-1}
    \def\decw{0.375}
    \def\dech{0.5}

    \def\gapx{0.5}
    \def\gapy{0.2}
    \def\xb{-0.5}

    \def\og{0.1}
    \def\slo{0.22}

    %% Alice's input line
    \draw[draw=black,very thick]
    (\chanx,\chany) --
    (\chanx + \slope*\chany - \slope*\Lone ,\Lone) --
    (\phiback,\Lone) --
    (\phiback-0.2,0.5*\Lone+0.5*\ency) --
    (\phiback,\ency) -- (\encx,\ency);

    %% Label of the channel input
    \node at (\phiback+0.2,\Lone+\slo)
    {$\ao$};

    \node at (\chanx + 0.1+ \slope*\chany - \slope*\Lone,\Lone+\slo)
    {$\ao$};

    %% Label of Alice's distiller input
    \node at (\phiback+0.2,\ency+\slo)
    {$\aoc$};

    % purified Alice distiller input state line and label
    \draw[draw=red] (\encx-\encw-0.1,\ency-0.2) --
    (\encx-\encw-0.1,\Lone+0.2);
    \node at (\encx-\encw+0.3,\Lone+0.3)
    {$\phi_{\ao\aoc}$};

    %%% overall input line (to the teleporter)
    \draw[draw=black,very thick]
    (\inxs,\iny) -- (\measx,\iny);
    %% Label of this line
    \node at (\inxs+0.2,\iny+\slo)
    {$\ai$};

    %%% E to meas line
    \draw[draw=black,very thick]
    (\encx,0.5*\ency+0.5*\measy) -- (\measx,0.5*\ency+0.5*\measy);
    %%% label for that line
    \node at (\encx+\encw+0.4,0.5*\ency+0.5*\measy+\slo)
    {$\aic$};

    %%% D to correcting U line
    \draw[draw=black,very thick]
    (\decx,\decy-\gapy) -- (\Ux,\decy-\gapy);
    %% Label of this line
    \node at (0.4*\decx+0.6*\Ux,\decy-\gapy+\slo)
    {$\bo$};

    %% forward classical communication
    \draw[draw=black, style=double]
    (\measx,\measy) -- (\measx+\gapx,\measy) --
    (\measx+\gapx-\slope*\Uy+\slope*\gapy+\slope*\measy,\Uy-\gapy) --
    (\Ux,\Uy-\gapy);

    %%% overall output line
    \draw[draw=black,very thick]
    (\Ux,\decy-\gapy) -- (\Ux+0.9,\decy-\gapy);
    %%% label for that line
    \node at (\Ux+0.7,\decy-\gapy+\slo)
    {$\bo$};

    %% channel output
    \draw[draw=black, very thick]
    (\chanx,\chany) --
    (\chanx+\slope*\chany-\slope*\decy-\slope*\gapy,\decy+\gapy) --
    (\decx,\decy+\gapy);

    %% Label of this line
    \node at (\decx-0.7,\decy+\gapy+\slo)
    {$\bi$};

    % "Bell" Measurement
    \draw [fill=white] (\measx-\measw,\measy-\meash)
    rectangle (\measx+\measw,\measy+\meash);
    \node (sc) at (\measx,\measy) {$\mc{M}$};
    % classical system label
    \node at (\measx+\measw+0.2,\measy+\slo)
    {$\C$};

    % Correcting unitary
    \draw [fill=white] (\Ux-\measw,\Uy-\meash)
    rectangle (\Ux+\measw,\Uy+\meash);
    \node (sc) at (\Ux,\Uy) {$\mc{U}$};

    % CHANNEL
    \draw [fill=white] (\chanx-\chanr,\chany-\chanr)
    rectangle (\chanx+\chanr,\chany+\chanr);
    \node (chan) at (\chanx,\chany) {$\mc{N}$};

    % internal distiller communication
    \draw[draw=black,thick,dashed]
    (\encx,\ency+\gapy) -- (\xb,\ency+\gapy) --
    (\xb+\slope*\ency-\slope*\decy+2*\slope*\gapy,\decy-\gapy) --
    (\decx,\decy-\gapy);

    % ENCODER
    \draw [fill=gray] (\encx-\encw,1.5)
    rectangle (\encx+\encw,0.5);
    \node (enc) at (\encx,\ency) {};

    % DECODER
    \draw [fill=gray] (\decx-\decw,\decy-\dech)
    rectangle (\decx+\decw,\decy+\dech);
    \node (dec) at (\decx,\decy) {};

    %% overall code outline
    \def\ya{0.4}
    \def\gapb{0.45}
    \draw[draw=black,dotted,thick]
    (\Ux+\decw+\og,\Uy-\dech-\og) --
    (\measx+\measw+
    \slope*\measy-\slope*\meash
    -\slope*\Uy+\slope*\dech,\Uy-\dech-\og) --
    (\measx+\measw,\measy-\meash-\og) --
    (\dleft,\measy-\meash-\og) --
    (\dleft,\Lone+0.5+\og) --
    (\xb+\gapb,\Lone+0.5+\og) --
    (\xb+\gapb,\ya)--
    (\xb+\gapb+\slope*\ya-\slope*\decy,\decy) --
    (\decx-\decw-\og,\decy) --
    (\decx-\decw-\og,\decy+\dech+\og) --
    (\Ux+\decw+\og,\decy+\dech+\og) --
    (\Ux+\decw+\og,\Uy-\dech-\og);
\end{tikzpicture}

%% file: code2dist.tex
\begin{tikzpicture}[scale=1]
\def\Lone{1.8}
\def\Lfour{0.7}
\def\Lthree{-0.75}
\def\lowLine{-1}
\def\slope{0.3}
\def\xa{0.2}

\def\chanx{-4.8}
\def\chany{\lowLine}
\def\chanr{0.375}

\def\iny{0.25}
\def\inxs{-5.7}

\def\encx{-2.6}
\def\ency{1}
\def\encw{0.375}

\def\measy{0.5}
\def\measx{-1.3}
\def\meash{0.4}
\def\measw{0.3}

\def\Ux{0.1}
\def\Uy{\lowLine}

\def\decx{1.4}
\def\decy{\lowLine+0.2}
\def\decw{0.375}
\def\dech{0.5}

\def\gapx{0.5}
\def\gapy{0.2}
\def\xb{-0.2}
\def\og{0.1}
\def\dleft{-4.3}
\def\slo{0.22}
\def\phiback{-3.5}
\def\rightx{2.4}

%% Alice's output line
\draw[draw=black,very thick]
(\rightx,\Lone) --
(\phiback,\Lone) --
(\phiback-0.2,0.5*\Lone+0.5*\ency) --
(\phiback,\ency) -- (\encx,\ency);

\node at (\rightx-0.3,\Lone+\slo)
{$\aic$};

\node at (\phiback+0.2,\Lone+\slo)
{$\aic$};

%% Label of Alice's distiller input
\node at (\phiback+0.2,\ency+\slo)
{$\ai$};

%%% overall input line (to the teleporter)
\draw[draw=black,very thick]
(\rightx,\lowLine) --
(\inxs,\lowLine) --
(\inxs - 0.2, 0.5*\lowLine+0.5*\iny) --
(\inxs,\iny) -- (\measx,\iny);

%% Labels of this line
\node at (\inxs+0.2,\iny+\slo)
{$\aoc$};

\node at (\chanx+\chanr+0.1,\iny+\slo)
{$\aoc$};

\node at (\inxs+0.2,\lowLine+\slo)
{$\aocc$};

\node at (\chanx+\chanr+0.2,\lowLine+\slo)
{$\bi$};

\node at (\rightx-0.2,\lowLine+\slo)
{$\bo$};

% purified input state line and label
\draw[draw=red] (\encx-\encw-0.1,\ency-0.2) --
(\encx-\encw-0.1,\Lone+0.2);
\node at (\encx-\encw+0.2,\Lone+0.3)
{$\phi_{\aic\ai}$};

%%% E to meas line
\draw[draw=black,very thick]
(\encx,0.5*\ency+0.5*\measy) -- (\measx,0.5*\ency+0.5*\measy);
%%% label for that line
\node at (\encx+\encw+0.4,0.5*\ency+0.5*\measy+\slo)
{$\ao$};

%% forward classical communication
\draw[draw=black, style=double]
(\measx,\measy) -- (\measx+\gapx,\measy) --
(\measx+\gapx-\slope*\Uy-\slope*\gapy+\slope*\measy,\Uy+\gapy) --
(\Ux,\Uy+\gapy);

%% Label of this line
\node at (\decx-0.7,\lowLine+\slo)
{$\bi$};

% "Bell" Measurement
\draw [fill=white] (\measx-\measw,\measy-\meash)
rectangle (\measx+\measw,\measy+\meash);
\node (sc) at (\measx,\measy) {$\mc{M}$};
% classical system label
\node at (\measx+\measw+0.2,\measy+\slo)
{$\C$};

% Correcting unitary
\draw [fill=white] (\Ux-\measw,\Uy-\meash)
rectangle (\Ux+\measw,\Uy+\meash);
\node (sc) at (\Ux,\Uy) {$\mc{V}$};

% CHANNEL
\draw [fill=white] (\chanx-\chanr,\chany-\chanr)
rectangle (\chanx+\chanr,\chany+\chanr);
\node (chan) at (\chanx,\chany) {$\mc{N}$};

% internal code communication
\draw[draw=black,thick,dashed]
(\encx,\ency+\gapy) -- (\xb,\ency+\gapy) --
(\xb+\slope*\ency-\slope*\decy+0*\slope*\gapy,\decy+\gapy) --
(\decx,\decy+\gapy);

% ENCODER
\draw [fill=gray] (\encx-\encw,1.5)
rectangle (\encx+\encw,0.5);
\node (enc) at (\encx,\ency) {};

% DECODER
\draw [fill=gray] (\decx-\decw,\decy-\dech)
rectangle (\decx+\decw,\decy+\dech);
\node (dec) at (\decx,\decy) {};

%% dotted outline of overall distiller
\draw [draw=black,thick,dotted]
(\phiback-0.2-\og,\Lone+\og+0.5) rectangle
(\decx+\decw+\og,\chany-\chanr-\og);

%% the entanglement for the Choi state
\draw[draw=red]
(\chanx-\chanr-0.1,\lowLine-0.3) --
(\chanx-\chanr-0.1,\measy+0.4);
\node at (\chanx-\chanr+0.3,\measy+0.5)
{$\phi_{\aocc\aoc}$};

%% the Choi state
\draw[draw=red]
(\chanx+\chanr+0.4,\lowLine-0.7) --
(\chanx+\chanr+0.4,\measy);
\node at (\chanx+\chanr+0.8,\lowLine-0.8)
{$\nu_{\bi\aoc}$};

\end{tikzpicture}